\documentclass[%
superscriptaddress,
twocolumn,
 nofootinbib,
 amsmath,
 amssymb,
 aps,
 prx,
]{revtex4-2}

\usepackage[T1]{fontenc}
\usepackage[utf8]{inputenc}
\usepackage{graphicx}%
\usepackage{dcolumn}%
\usepackage{bm}%
\usepackage[dvipsnames]{xcolor}
\usepackage[colorlinks=true,linkcolor=teal,citecolor=Maroon,urlcolor=Maroon]{hyperref}
\usepackage{physics}
\usepackage{siunitx}
\usepackage{float}
\usepackage{tikz}
\usetikzlibrary{quantikz}
\usepackage{mathbbol}
\usepackage{enumitem}
\usepackage{lucas}
\usepackage{times}

\newcommand{\qgate}[1]{\textsc{#1}}

\newcommand{\qcode}[1]{\texttt{#1}}

\newcommand{\secref}[1]{Section~\hyperref[#1]{\ref*{#1}}}
\newcommand{\appref}[1]{Appendix~\hyperref[#1]{\ref*{#1}}}
\newcommand{\tabref}[1]{Table~\hyperref[#1]{\ref*{#1}}}
\newcommand{\figref}[1]{Fig.~\hyperref[#1]{\ref*{#1}}}
\newcommand{\sfigref}[2]{Fig.~\hyperref[#1]{\ref*{#1}(#2)}}

\newcommand{\etal}{et al.~}

\usepackage{svg}
\newcommand{\orcid}[1]{\href{https://orcid.org/#1}{\includesvg[width=8pt]{orcid_icon}}}

\begin{document}

\title{Analog information decoding of bosonic quantum LDPC codes}

\author{Lucas Berent}
\email{lucas.berent@tum.de}
\thanks{These authors contributed equally.}
\affiliation{Technical University of Munich, Germany}

\author{Timo Hillmann}
\email{timo.hillmann@rwth-aachen.de}
\thanks{These authors contributed equally.}
\affiliation{Department of Microtechnology and Nanoscience (MC2), Chalmers University of Technology, SE-412 96 Gothenburg, Sweden}

\author{Jens Eisert}
\email{jense@zedat.fu-berlin.de}
\affiliation{Freie Universit\"at Berlin, Germany}
\affiliation{Helmholtz-Zentrum Berlin f\"ur Materialien und Energie, Germany}

\author{Robert Wille}
\email{robert.wille@tum.de}
\affiliation{Technical University of Munich, Germany}
\affiliation{Software Competence Center Hagenberg, Austria}

\author{Joschka Roffe}
\email{joschka@roffe.eu}
\affiliation{Quantum Software Lab, University of Edinburgh}
\affiliation{Freie Universit\"at Berlin, Germany}

\begin{abstract}
Quantum error correction is crucial for scalable quantum information processing applications. Traditional discrete-variable quantum codes that use multiple two-level systems to encode logical information can be hardware-intensive. An alternative approach is provided by bosonic codes, which use the infinite-dimensional Hilbert space of harmonic oscillators to encode quantum information. Two promising features of bosonic codes are that syndrome measurements are natively analog and that they can be concatenated with discrete-variable codes. In this work, we propose novel decoding methods that explicitly exploit the analog syndrome information obtained from the bosonic qubit readout in a concatenated architecture. Our methods are versatile and can be generally applied to any bosonic code concatenated with a quantum low-density parity-check (QLDPC) code. Furthermore, we introduce the concept of quasi-single-shot protocols as a novel approach that significantly reduces the number of repeated syndrome measurements required when decoding under phenomenological noise. To realize the protocol, we present the first implementation of time-domain decoding with the overlapping window method for general QLDPC codes and a novel analog single-shot decoding method. Our results lay the foundation for general decoding algorithms using analog information and demonstrate promising results in the direction of fault-tolerant quantum computation with concatenated bosonic-QLDPC codes.
\end{abstract}

\maketitle

\section{Introduction}

For quantum computing, an important design consideration is the choice of technology used to physically realize qubits. 
A standard approach is to engineer \emph{discrete variable} (DV) qubits, where the basis states are defined by two distinct energy levels of the quantum system. 
Examples include ion-trap qubits~\cite{bruzewicz_trapped-ion_2019, wineland_nobel_2013}, superconducting qubits~\cite{devoret_superconducting_2013}, and spin qubits~\cite{burkard_semiconductor_2023, henriet_quantum_2020}.
However, discrete variable qubits pose challenges, primarily in the intricate task of effectively isolating the two-level encoding from external influences that introduce errors. There is substantial evidence that, in the absence of error suppression methods, 
the coherence of quantum circuits is limited to at most a logarithmic depth \cite{deshpande_tight_2022, stilck_franca_limitations_2021}, which represents a challenge for near-term quantum computing. 
To suppress qubit errors, we %
can resort to notions of 
\emph{quantum error correction} (QEC).
The goal of QEC is to harness entanglement to redundantly encode quantum information in the logical qubit state of a larger physical Hilbert space, albeit to the extent of substantial overhead. 
The QEC encoding provides the system with additional degrees of freedom that can be used to detect and correct errors in real time. 
Beyond the initial encoding, full QEC protocols must incorporate logical gates that allow the encoded quantum information to be manipulated in a fault-tolerant way. 
This provides the ability to compute fault-tolerantly, whilst keeping the system protected against local noise. Moreover, given a (potentially corrupted) encoded state, a central task is to check whether errors occurred on the encoded information, and if so to decode, i.e., to computationally derive a suitable recovery operation to restore an error-free state. 
Identifying and realizing feasible and practical schemes for fault-tolerant operation remains one of the core challenges of quantum computing, both in experiment and in theory.

On the highest level, both discrete variable and \emph{continuous variable} (CV) codes for quantum error correction have been considered in the literature and are seen as candidates for 
feasible quantum codes. The latter offer some compelling advantages for a number of platforms, notably for \emph{cat codes}~\cite{cochrane_macroscopically_1999} and \emph{Gottesman-Kitaev-Preskill} (GKP) codes~\cite{gottesman_encoding_2001}.
However, there are also some challenges that are unique to the continuous
variable setting. In particular, it is not obvious how to do decoding in light of continuous syndrome information, as it is not clear how to make use of continuous information in this task. For example, the development of good decoders for GKP codes constitutes a well-known technical challenge. The lack of good methods of decoding for quantum error-correcting codes with a continuous component can be
seen as a roadblock in the field.

In this work, we report substantial progress on the partial use of continuous
syndrome information in the notions of quantum error correction. In this way, we aim
to bring the advantages of DV and CV quantum error correction closer together.
To this end, we
explore the combination of two promising classes of codes, instances of
bosonic codes and quantum \emph{low-density parity-check} (LDPC) codes, and investigate 
suitable decoding protocols for general bosonic-LDPC constructions that make such
partial use of analog information.
To further motivate this combination and to show how
continuous syndrome information comes into play,
we give a brief overview in the following.

\paragraph{Bosonic qubits.}
Bosonic encodings offer an alternative to 
DV qubits~\cite{ofek_extending_2016, sivak_real-time_2023}. 
In this approach, the logical qubit state is non-locally embedded within the phase space of an infinite-dimensional quantum harmonic oscillator. 
The principal advantage of using bosonic qubits is that the infinite-dimensional Hilbert space provides the redundancy needed to correct physical oscillator errors. 
Consequently, bosonic encodings can be interpreted as intrinsic QEC protocols at the individual qubit level. 
In principle, this provides an efficient route to fault tolerance, and bosonic codes have been extensively explored using protocols such as GKP codes~\cite{gottesman_encoding_2001} and cat codes~\cite{cochrane_macroscopically_1999}.

Similarly to discrete variable qubits, multiple bosonic qubits can be combined via a QEC code to create a logical state. This strategy is usually referred to as a \textit{concatenated} bosonic code: at the inner level, the individual qubits are protected by their bosonic encoding, and at the outer level, the bosonic qubits are wired together to form a logical state. Concatenated codes therefore combine the benefits of both bosonic and discrete variable QEC codes.
 
Another appealing feature of various bosonic codes is that they can be precisely tuned to exhibit noise asymmetry in their qubit-level error model. For example, recent experiments have shown that cat code qubits can be engineered to have phase-flip rates that dominate over bit-flips by almost three orders of magnitude~\cite{grimm_stabilization_2020, lescanne_exponential_2020}. In a concatenated bosonic code, highly biased inner-level qubits can reduce the overhead required by the outer code. For example, in recent work by Darmawan~\etal\cite{darmawan_practical_2021}, it is proposed that cat code qubits can be concatenated with the XZZX code~\cite{bonilla_ataides_xzzx_2021}, an instance of a surface code modified by Clifford conjugations that has extremely high thresholds in the limit of large bias~\cite{tiurev_correcting_2023}. 

\paragraph{Concatenated bosonic codes.}
To date, most studies of bosonic codes have focused on their concatenation with repetition codes~\cite{guillaud_repetition_2019, guillaud_error_2021, chamberland_building_2022, regent_high-performance_2023, stafford_biased_2023, guillaud_repetition_2019} or two-dimensional topological codes~\cite{vuillot_quantum_2019, noh_fault-tolerant_2020, darmawan_practical_2021, raveendran_finite_2022}. 
Such codes are favored for near-term experiments because they require only nearest-neighbor interactions. 
This facilitates their implementation on a two-dimensional array of qubits, making them particularly suitable for architectures such as superconducting qubits. 
However, from an information theoretical standpoint, two-dimensional topological codes %
may be seen as being far from optimal. The main drawback lies in their poor rate: the surface code, for instance, encodes only a single logical qubit per patch. 
This means that increasing the surface code distance comes at the expense of the encoding density. 
This is in stark contrast to the efficiency of contemporary classical error correction protocols. 
In particular, numerous classical communication technologies employ %
LDPC codes. 
Such codes have the advantage of preserving a constant encoding density even as the code distance is scaled. 
Moreover, it has been shown that in the asymptotic limit, LDPC codes can approach the Shannon bound, which represents the theoretical limit on the rate of information transfer through a noisy channel.

\paragraph{Low-density parity-check quantum codes.}
Until recently, it was an open question as to whether 
\emph{quantum LDPC} (QLDPC) codes with \emph{good} parameter scaling comparable to their classical counterparts exist. 
This question has recently been answered in the affirmative via a series of theoretical breakthroughs~\cite{hastings_fiber_2021, breuckmann_balanced_2021, panteleev_asymptotically_2022, dinur_good_2022, leverrier_quantum_2022}. 
Central to these innovations has been the use of sophisticated product constructions that provide procedures for transforming classical LDPC codes into quantum codes. 
The resulting QLDPC codes exhibit constant encoding rates and distances that scale proportionally with the length of the code.

Implementing QLDPC codes poses greater challenges than their planar counterparts. Several no-go theorems indicate that implementing (good) QLDPC codes will not be possible in two-dimensional geometrically local architectures~\cite{bravyi_no-go_2009,  baspin_connectivity_2022, baspin_quantifying_2022, baspin_improved_2023, baspin_lower_2023}: they must necessarily involve geometrically non-local connections. 
Nonetheless, various qubit technologies are being developed that enable long-range interactions needed to implement QLDPC codes~\cite{xu_constant-overhead_2023, ramette_any--any_2022,barredo_synthetic_2018, omran_generation_2019, strikis_quantum_2023}. 
In this setting, QLDPC codes promise quantum computation with considerably reduced overhead compared to the surface code~\cite{bravyi_high-threshold_2023, xu_constant-overhead_2023}.

In Ref.~\cite{raveendran_finite_2022}, Raveendran~\etal{}have explored the concatenation of GKP bosonic qubits with QLDPC codes based on the lifted product construction~\cite{panteleev_degenerate_2021}.
Their results show that there are distinct advantages to using a concatenated bosonic code over directly implementing a QLDPC code with discrete variable qubits. 
Specifically, it is possible to feed forward analog information from the GKP qubit readout to improve the performance of the outer code's decoder. This leads to improved error suppression beyond what is achievable with discrete variable qubits alone.

In light of this promising line of research on concatenated bosonic codes, we focus on QEC protocols constructed from bosonic-LDPC codes in a general fashion, assuming only that the inner code provides analog syndrome readout and possibly a noise bias and that the outer code is an arbitrary quantum LDPC code.
We focus on the decoding of such codes and demonstrate that, by combining their properties, we obtain high-performance QEC protocols.

\subsection{Overview of contributions}

\begin{figure*}[t]
    \centering
    \includegraphics{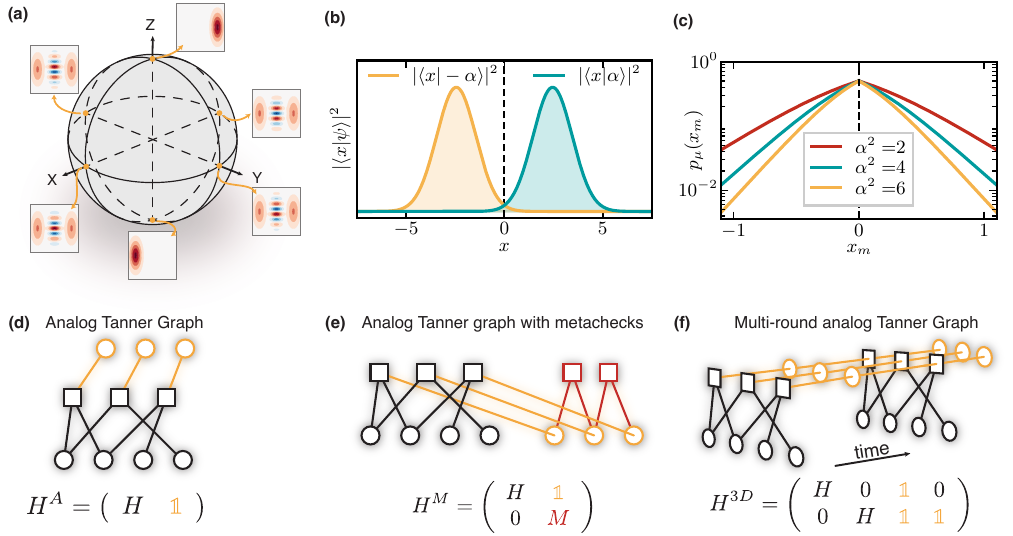}
    \caption{Overview of our main techniques. 
    \textbf{(a)}~We investigate QLDPC codes concatenated with cat qubits encoded in coherent states of a harmonic oscillator.
    \textbf{(b)}~An important property of cat qubits, in addition to their biased noise model, is that the syndrome information obtained from qubit readout is intrinsically analog-valued, as the wavefunction of a coherent state is a Gaussian centered at $\alpha$.
    \textbf{(c)}~Depending on the measured value $x_m$ during (quadrature) readout, we can assign an outcome-dependent error probability $p(x_m)$ that is a function of the size of the cat qubit $\alpha^2$.
    \textbf{(d)}~We incorporate the analog information obtained during the syndrome measurements into a Tanner graph construction that we refer to as an \emph{analog Tanner graph} (ATG). The ATG stores the analog syndrome information directly in the factor graph used for decoding.
    We show that the ATG construction can be adapted to work with decoding strategies such as \textbf{(e)}~single-shot shot decoding and \textbf{(f)}~overlapping window time-domain decoding.
    }
    \label{fig:fig1}
\end{figure*} 

In this work, we develop methods for the decoding and analysis of concatenated bosonic-LDPC codes. 
To this end, our primary test bed is a protocol in which cat code qubits (inner code) are concatenated with the three-dimensional surface code (outer code). We choose to focus on the three-dimensional surface code as it is perhaps the simplest example of a QLDPC code that extends upon the capabilities of the standard two-dimensional surface code. 
From an implementation perspective, three-dimensional surface codes can be realized with relatively few long-range connections in two and completely locally in three dimensions~\cite{vasmer_three-dimensional_2019}. 
Suitable experimental platforms for implementing three-dimensional codes include architectures with photonic links or neutral atom arrays~\cite{barredo_synthetic_2018, omran_generation_2019}. 
For comparison, the concatenated lifted product schemes of Ref.~\cite{raveendran_finite_2022} require arbitrary qubit connectivity in general. 
Another distinguishing feature of the three-dimensional surface code is the fact that it has a transversal CCZ gate~\cite{vasmer_three-dimensional_2019}. 
This leaves it fully equipped for universal quantum computation without the need for resource-intensive magic state injection. 

We introduce a novel decoding method, \textit{analog Tanner graph decoding} (ATD), which makes use of the analog readout information from the bosonic qubits. 
Our numerical simulations show that this leads to improved thresholds for decoding with the three-dimensional surface codes. Furthermore, we demonstrate that using ATD, the number of repetitions required by the non-single-shot component of the three-dimensional surface codes can be reduced to a small number.
In this setting, we refer to three-dimensional surface codes as being \textit{quasi-single shot}.
Finally, this paper is accompanied by open-source software tools that facilitate the reproduction and extension of our analysis of concatenated bosonic-LDPC codes and provide means to automatically conduct respective numerical simulations.
Our results are summarized in more detail below:

\paragraph{Analog Tanner graph decoding.} 
Fundamentally, in discrete-variable QEC, the syndromes obtained from stabilizer measurements are discrete, although readout techniques can yield an analog value. 
This stands in contrast to bosonic qubits, where measurements yield analog outcomes due to the infinite-dimensional Hilbert space. 
The strength of the analog readout can be used to assign an uncertainty associated with the measurement. 
Syndromes derived from analog bosonic readout are termed \textit{analog syndromes}.
The \emph{analog Tanner graph decoder} (ATD) we introduce in this work provides a method for mapping analog syndromes to a belief propagation decoder, sketched in~\figref{fig:fig1}{d}. 
Our approach is extremely versatile: ATD can be applied to any stabilizer code with analog syndrome information.
Furthermore, it is possible to incorporate ATD as part of both single-shot and time-domain decoding strategies, as illustrated in~\figref{fig:fig1}{e} and~\figref{fig:fig1}{f}.

\paragraph{Single-shot decoding with analog information.} 
A problem in QEC is that syndromes must be extracted using auxiliary qubits that are also susceptible to errors. 
As such, there is uncertainty associated with any syndrome we measure. 
To counteract this problem, we can adapt the \textit{time-domain} approach in which syndrome measurements are repeated $\propto d$ times, where $d$ is the code distance. 
Measurement errors can then be accounted for by considering the entire syndrome history in the decoding. 
An alternative strategy is to use a code that has the so-called \textit{single-shot} property. 
Single-shot codes have an additional structure that allows measurement errors to be directly corrected, removing the need to decode over time~\cite{bombin_single-shot_2015}. 
The three-dimensional surface code is single-shot for phase noise but requires time-domain decoding for bit-flip noise~\cite{quintavalle_single-shot_2021, vasmer_three-dimensional_2019}. Our numerical results show that under ATD decoding, the single-shot component of the concatenated three-dimensional surface code has a sustained threshold of $9.9\%$ under phenomenological noise. 
This improves over the previous best-observed threshold for the discrete variable three-dimensional surface codes of $7.1\%$~\cite{higgott_improved_2023}.

\paragraph{Quasi-single-shot protocol.}

Time-domain decoding can lead to large overheads for quantum codes, since the number of repeated syndrome measurements required is proportional to the code distance $d$~\cite{shor_fault-tolerant_1996, delfosse_beyond_2022, dennis_topological_2002}. 
This increases the length of the decoding cycle and reduces the frequency with which logical operations can be applied~\cite{skoric_parallel_2022}. 
To address this overhead, we propose a novel protocol called ($w$)-\emph{quasi-single-shot} decoding. 
The key idea is to use the analog information obtained during the syndrome readout of the inner bosonic code to enhance the decoding performance and reduce the need to repeat the measurement multiple times in the noisy readout setting.
On an intuitive level, this idea is somewhat analogous to the fact that to learn $m$ bits during the quantum phase estimation algorithm, one needs $O(m)$ measurements~\cite{kitaev_quantum_1995}, whereas, using bosonic qubits, one can learn the phase (in principle) using a single measurement~\cite{gottesman_encoding_2001}.

Quasi-single-shot decoding leverages our strategies for decoding with analog information, and we obtain a scheme that reduces the required number of repeated syndrome measurement rounds to a small number $w \ll d$.
In combination with the tunable noise bias of the inner cat qubit code, the quasi-single-shot three-dimensional surface code yields a three-dimensional topological code with significantly higher thresholds and reduced time-overhead when compared to two-dimensional topological codes.

\paragraph{Open-source software tools.}
We have implemented a set of software tools as part of the Munich Quantum Toolkit (MQT) to foster further research in the direction and to provide the community with numerical tools to conduct simulations and reproduce our results.

\smallskip
The remainder of this work is structured as follows. In~\secref{sec:background}, we review the main concepts of bosonic quantum codes and quantum error correction that are needed throughout this work. 
In~\secref{sec:decoding-main}, we discuss iterative decoding approaches that are fundamental for state-of-the-art algorithms and the proposed techniques.
Then, in~\secref{sec:atd} our decoding strategies for decoding quantum codes using analog information and the respective numerical simulation results are presented. 
In~\secref{sec:qss-codes}, we elaborate on the proposed quasi-single-shot protocol and present numerical results.
\secref{sec:architectures} focuses on more practical aspects around the considered bosonic-LDPC code architectures and reviews important open challenges with respect to potential implementations in superconducting architectures.
We conclude and give a brief overview of future work in~\secref{sec:conclusion}.

\section{Preliminaries}\label{sec:background}
Here, we introduce basic notions of quantum coding that are needed throughout the rest of this work. We assume the reader is familiar with fundamental notions of quantum information theory and quantum error correction.

\subsection{Bosonic quantum codes}

The use of bosonic codes toward achieving fault-tolerant quantum computing has become a promising alternative to the approach based on so-called \emph{discrete-variable} qubits such as spin qubits, trapped ions, neutral atoms, or superconducting qubits~\cite{bruzewicz_trapped-ion_2019, wineland_nobel_2013, devoret_superconducting_2013, burkard_semiconductor_2023, henriet_quantum_2020}.
It is often quoted, see, for instance, Refs.~\cite{joshi_quantum_2021, cai_bosonic_2021} that the advantage of bosonic codes over discrete-variable codes is the fact that a single bosonic mode, realized, e.g., in a superconducting cavity or the motional states of trapped ions, lives in an infinite-dimensional Hilbert space. 
Bosonic codes therefore offer a hardware-efficient route toward fault-tolerant quantum computing, since the principle of quantum error correction relies on encoding logical information redundantly in a subspace of a much larger Hilbert space, as opposed to using multiple two-level systems.
The main idea of how bosonic codes can be realized in practice is that a particular bosonic code constitutes the first layer in a concatenated quantum error correction code, consisting of (at least) one continuous-variable and one discrete-variable code, which is called \emph{outer code} in this context.
As there are various realizations and families of discrete-variable codes, there are similarly multiple families of bosonic codes. 
As the choice of the discrete-variable code determines properties, such as the availability of transversal gates, the choice of a particular bosonic code can be tailored to the physical platform on which the code is realized, but also determines the (effective) noise model that is relevant for the outer code as well.

As such, in addition to their large surrounding Hilbert space, bosonic codes offer further advantages over two-level systems relevant to the design of quantum error-correcting codes and decoders.
One of these features is that many families of bosonic codes have a biased-noise error model, where some types of errors occur more frequently than others. 
While discrete-variable qubits can have a biased noise channel, this noise bias cannot be preserved by gate operations needed throughout the QEC protocol, i.e.,  there does not exist a bias-preserving \qgate{CNOT} implementation~\cite{guillaud_repetition_2019}.
While a Pauli \qgate{Z} error occurring before the gate is not converted to other types of Pauli errors, a \qgate{Z} error during the execution of the gate will be converted to other types of Paulis.
Thus, intuitively, for the \qgate{CNOT} $= \qgate{CX}(\pi)$ one must require that \qgate{Z} errors commute with the gate at \emph{all times}, i.e.,  $\comm*{Z}{\qgate{CX}(\alpha)} = 0, \, \forall \alpha \in [0, \pi]$.
However, implementation of the \qgate{CNOT} in such a way is not possible without leaving the code space~\cite{guillaud_repetition_2021}.
Hence, the bias is annihilated during the computation and bias-tailored QEC codes such as the ones proposed in Refs.~\cite{tuckett_ultrahigh_2018, tuckett_fault-tolerant_2020, roffe_bias-tailored_2023, huang_tailoring_2022} are not beneficial for discrete-variable systems.
However, the additional degrees of freedom of continuous-variable states natively allow for bias-preserving operations~\cite{guillaud_repetition_2019, puri_bias-preserving_2020}.

While the additional degrees of freedom of bosonic codes yield a clear advantage over two-level systems, the continuous support of states in quantum phase space has the consequence that any measurement, e.g., the measurement of a stabilizer check, of such a state is inherently imprecise.
Instead of dismissing this characteristic of bosonic quantum codes as a drawback, it can equivalently be seen as a feature that yields additional information during the qubit readout, called \emph{analog information} that can be used for decoding.
We note that in practice the readout of a two-level system in finite time also yields a continuous outcome. 
However, with technological progress in recent years, the readout uncertainty due to the finite measurement time has become almost negligible~\cite{swiadek_enhancing_2023}.
For the interested reader, we describe the properties of stabilized cat codes more explicitly in \secref{sec:architectures}.

\subsection{Low-density parity-check quantum codes}
In the following, all vector spaces are over $\Ft{}$ unless stated otherwise. We use $[\ell]$ to denote the set $\set{1,2,\dots,\ell}$.
A (discrete variable, DV) $[[n,k,d]]$-quantum \emph{stabilizer code} is defined by an Abelian subgroup $S$ of the $n$-qubit Pauli group $\calP_n$ that does not contain $-I$. The generators of $S$ are commonly called stabilizer \emph{checks} of the code. 

An important class of quantum stabilizer codes are \emph{Calderbank-Shor-Steane} (CSS) codes. The defining feature of CSS codes is that their stabilizer generators can be split into two decoupled sets $S_X$ and $S_Z$, which contain only products of the Pauli \qgate{X} and Pauli \qgate{Z} operators, respectively.
By using the (isomorphic) mapping between the $n$-qubit Pauli group (modulo global phases) and binary vector spaces
\cite{gottesman_stabilizer_1997}
\begin{equation}
\calP_n/\set{\pm I, \pm iI} \cong \Ft^{2n},
\end{equation}
it is possible to represent the $S_X$ and $S_Z$ stabilizers of a CSS code as two matrices: $H_X\in \Ft^{r_X \times n}$ and $H_Z\in \Ft^{r_Z\times n}$. 
These matrices can be interpreted as the parity-check matrices of two classical linear codes, the first designed to correct bit-flips and the second to correct phase-flips. For any CSS code, the $H_X$ and $H_Z$ matrices must satisfy the following orthogonality condition.
\begin{equation}
    \label{eq:css-orthogonality}
	H_Z \cdot H_X^{\top} = 0 \equiv H_X \cdot H_Z^{\top} = 0,
\end{equation}
which ensures that the \qgate{X} and \qgate{Z} stabilizer generators commute, i.e., their supports have even overlap.
A CSS code can then be defined as a code over $\Ft^{2n}$ with check matrix
\begin{equation}
    H = \begin{pmatrix}
        0 & H_Z\\
        H_X & 0
    \end{pmatrix}\rm.
\end{equation}
The CSS \emph{syndrome} $s$ of a qubit error $e=(e_X, e_Z)$ is defined as
\begin{equation}\label{eq:syndrome}
    s = (s_X,s_Z) = (H_Z\cdot e_X , H_X \cdot e_Z)\rm.
\end{equation}
From the syndrome equations above, it is clear that the decoding of CSS codes amounts to two independent classical decoding problems for phase-flips and bit-flips, respectively. 
However, note that decoding the two check sides independently ignores correlations, such as Y-errors, and thus is suboptimal in general.
In the following, we will drop the subscripts, but it should be assumed that we are referring to a single error type ($X$ or $Z$) unless otherwise stated.

In the language of vector spaces, the goal of \emph{decoding} is, given a syndrome, to infer an estimate $\varepsilon$ s.t. $s = H\cdot\varepsilon$ that can be used to apply a \emph{recovery operation} to restore an error-free logical state. 
The decoding is successful if the \emph{residual error} $r = e+\varepsilon$ is a stabilizer, and it fails if $r$ is a logical operator. 

The \emph{Tanner graph} $\calT(H) = (V_D \cup V_C, E)$ of a code defined by the parity-check matrix $H$ is a bipartite graph with data nodes $V_D$ and check nodes $V_C$ whose incidence matrix is $H$. 
That is, given $H$, $\calT(H)$ can be constructed by creating a data node $d\in V_D$ for each column and a check node $c \in V_C$ for each row of $H$, and inserting an edge $e=(v_c,v_d)\in E$ if $H_{(c,d)} = 1$.

\begin{example}[Tanner graph of the Hamming code]
    Consider the Hamming code, defined as the kernel of the parity-check matrix 
    \[H = \left(\begin{matrix}
        1 & 0 & 0 & 1 & 0 & 1 & 1 \\
        0 & 1 & 0 & 1 & 1 & 0 & 1 \\
        0 & 0 & 1 & 0 & 1 & 1 & 1
    \end{matrix}
    \right).\]
    The corresponding Tanner graph $\calT(H)$ is illustrated in~\figref{fig:tannergraph-ham-ex}.
    \begin{figure}[!b]
        \centering
        \includegraphics{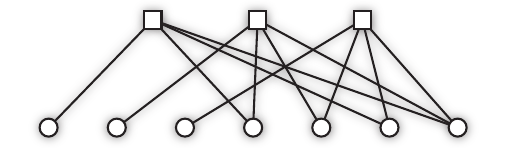}
        \caption{Tanner graph of the Hamming code. The square nodes represent the checks $V_C$ (rows of the check matrix $H$) and the circles represent data nodes $V_D$ (columns of $H$).}
        \label{fig:tannergraph-ham-ex}
    \end{figure}
\end{example}
In the context of decoding algorithms based on factorization of probability distributions, the Tanner graph is frequently referred to as \emph{factor graph}.

A family of \emph{quantum low-density parity-check} (QLDPC) codes refers to a stabilizer code family with the additional property that the stabilizer generators are sparse. 
More specifically, it is required that the degree of the data and check nodes in the Tanner graph is upper bounded by a constant (independent of the code size). 
Intuitively, each check has a constant weight and each qubit participates in a constant number of checks.

\subsubsection{Three-dimensional surface codes}\label{sec:prelim_3dsc}
\begin{figure*}[t]
    \centering
    \includegraphics{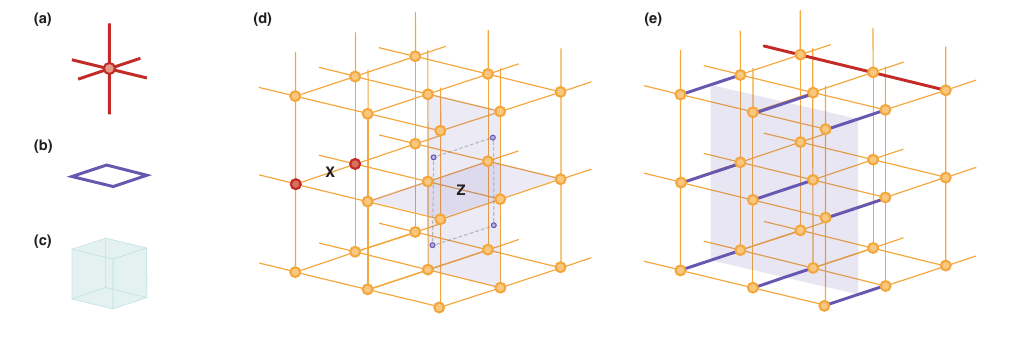}
    \caption{three-dimensional surface codes with periodic boundaries indicated by additional edges on the sides of the lattice. \textbf{(a)} A vertex check (of weight six). \textbf{(b)} A face check (of weight four). \textbf{(c)} A volume check (of weight 12). \textbf{(d)} three-dimensional lattice $\Lambda$ with open boundaries. Rough boundaries are indicated with open edges. A single qubit \qgate{X} error gives a pair-like syndrome at the endpoints of the error string (indicated by red vertices). A single qubit \qgate{Z} error produces a loop-like syndrome at adjacent faces (indicated by blue faces). The loop can be readily seen in the dual lattice pictures whose dual edges are indicated with dashed lines. \textbf{(e)} Logical operators of the three-dimensional surface codes. A loop-like string corresponds to a logical \qgate{X} operator. A logical \qgate{Z} operator corresponds to a loop of faces in the dual lattice, i.e., forming a \emph{dual sheet} wrapping across the lattice along two axes.}
    \label{fig:3dsc}
\end{figure*}
For the remainder of this work, we focus on \emph{topological quantum codes}~\cite{bombin_topological_2006, bravyi_quantum_1998,kitaev_fault-tolerant_2003} as the ``outer code'' of a concatenated code. Such codes are derived from a geometric $D$ dimensional lattice and have local connectivity in $D$-dimensional space. 
In particular, we focus on the \emph{three-dimensional surface code}  (3DSC)~\cite{dennis_topological_2002, hamma_string_2005, vasmer_three-dimensional_2019, huang_tailoring_2022, kubica_single-shot_2022} as a representative QLDPC code.
Each constituent is identified with a qubit with Hilbert space $\mathbb{C}^2$.
Consider a three-dimensional lattice in Euclidean space captured by a graph $\Lambda=(V,E)$ consisting of vertices $V$, edges $E$, faces $F$, and volumes $W$.
For two objects in the lattice $v,f$, we write, $v \sim f$ ($f\sim v$) to indicate that $v$ is adjacent to $f$.

To define a quantum code on the lattice $\Lambda$, we associate qubits with edges and checks with vertices and faces, i.e.,  we define \emph{vertex checks} $A_v, v\in V$, 
acting on edges adjacent to vertices
\begin{equation}\label{eq:3dsc-v-checks}
 A_v := \prod_{e\in E, e \sim v} \qgate{Z}_e \,
\end{equation}
depicted in~\sfigref{fig:3dsc}{a},
and \emph{face checks} $B_f, f\in F$ acting on edges adjacent to faces
\begin{equation}\label{eq:3dsc-f-checks}
 B_f := \prod_{e\in E, e \sim f} \qgate{X}_e \, 
\end{equation}
shown in cf.~\sfigref{fig:3dsc}{b}.
To reason about the logical operators of the code, let us introduce some informal notation.
For a more formal discussion in the language of homology, we refer the reader to~\appref{sec:app-homology}.

Let $\zeta$ be an edge path in $\Lambda$.
A \emph{string} operator is defined as a Pauli $X/Z$ 
operator whose support corresponds to the qubits that 
are associated to edges in $\zeta$ as
\begin{equation}
 S_\zeta^P:= \prod_{i\in \zeta} P_i, P\in \set{X,Z}.  
\end{equation}
The weight of a string is the size of its support, i.e., the number of non-trivial Paulis. 
In addition to the lattice $\Lambda$, consider also the dual lattice $\Lambda^*$, which is obtained from $\Lambda$ by associating dual vertices to volumes, dual edges to faces, dual faces to edges and dual volumes to vertices, and define \emph{co-string operators} as string operators on $\Lambda^*$. 
Hence, qubits are associated with dual faces, \qgate{Z}-checks to dual volumes, and \qgate{X}-checks to dual edges.

Errors correspond to \qgate{X}/\qgate{Z}-strings on the lattice that cause the anti-commuting (\qgate{Z}/\qgate{X}) adjacent checks to be ``flagged'', indicating an error occurred (all other, non-adjacent checks are not violated and thus not flagged).
The syndromes caused by violated vertex checks are created in pairs at the vertices that are the endpoints of \qgate{X}-error strings, as depicted in~\sfigref{fig:3dsc}{d} for a single-qubit error (i.e., a length-1 error string in $\Lambda$).
The syndrome of a \qgate{Z}-error string corresponds to the adjacent violated face-checks as shown in~\sfigref{fig:3dsc}{d} for a single qubit error. 
Syndromes caused by violated face checks are better illustrated in the dual lattice $\Lambda^*$ (recall that edges correspond to dual faces, so a string operator in $\Lambda$ corresponds to the respective collection of faces in $\Lambda^*$).
Considering $\Lambda^*$ it can be seen that \qgate{Z}-syndromes have a loop-like geometry in the lattice. 
In fact, the loop-like syndromes induced by face checks lead to an additional structure that has been shown to imply the single-shot property, see also~\secref{sec:ss-qldpc}.

Given a three-dimensional lattice, we distinguish two types of boundary conditions. 
When $\Lambda$ has periodic boundaries (depicted by additional edges in~\figref{fig:3dsc}), i.e., a tessellation of a three-dimensional torus (therefore also called the \emph{three-dimensional toric code}, 3DTC), the logical $\qgate{X}$ operators correspond to strings that form loops on the lattice along one axis, as depicted in~\figref{fig:3dsc}{e}. 
The logical $\qgate{Z}$ operators correspond to loop-like sets of faces, called \emph{sheets} in $\Lambda^*$, i.e., sets of faces that go across the dual lattice along two axes as illustrated in~\figref{fig:3dsc}{e}. 
Topologically, logical operators correspond to \emph{non-trivial} loops, i.e., loops (of edges and faces, respectively) that are non-contractible and the contractible loops (those that enclose a region on the lattice) correspond to stabilizers of the code. 
It is straightforward to see that there are three pairs of logical operators $\bar{\qgate{X}}, \bar{\qgate{Z}}$, corresponding to three (non-equivalent) minimum-weight loops, one along each of the three different axes, and the three corresponding orthogonal sheets. 
Therefore, the code encodes $k=3$ logical qubits. 
The \emph{weight} of a logical operator is the number of qubits in its support, i.e., the length of $\zeta$.

If we instead consider a code on a three-dimensional lattice $\Lambda$ with open boundaries, we define two opposite sides of the lattice as $X$-type boundaries, called \emph{smooth boundaries} and the four remaining sides as $Z$-type boundaries, called \emph{rough boundaries}---in analogy to the two-dimensional surface code. 
A string operator $S_\zeta^X$ (of minimal length) that connects the smooth boundaries is a logical \qgate{X} operator on the code, and a dual string operator $S_\zeta^Z$ connecting the four rough boundaries corresponds to a logical \qgate{Z} operator.
Since in the presence of open boundaries, there is only a single pair of such strings s.t. they are orthogonal and connect the respective boundaries, the code encodes a single logical qubit.

The X-distance $d_{X}$ of the code is defined as the minimum weight of a logical X operator, and analogously for $d_Z$. 
That is, $d_{X}$ is the minimum number of edges across $\Lambda$, and $d_Z$ is the minimum number of edges corresponding to a sheet across the dual lattice in the presence of periodic boundaries.

For a lattice with open boundaries, $d_X$ is defined as the minimum number of edges in a string connecting smooth boundaries and $d_Z$ is defined as the minimum number of edges corresponding to a sheet connecting rough boundaries.
For instance, the code depicted in~\figref{fig:3dsc} has distances $d_X = 3, d_Z = 9$. 
In summary, the \emph{three-dimensional surface code} (3DSC) 
parameters (with open and periodic boundaries) are given by
\begin{align}
    \text{3DSC: } & \left[\left[2L(L-1)^2+L^3, 1, d_X=L, d_Z=L^2\right]\right],\nonumber \\
    \text{3DTC: } & \left[\left[3L^3, 3, d_X=L,d_Z=L^2\right]\right] .\nonumber 
\end{align}
Note that by associating \qgate{X} checks with vertices and \qgate{Z} checks with faces, we obtain an equivalent code, where the corresponding notions of logicals and distances are simply exchanged.

\subsubsection{Single-shot codes}\label{sec:ss-qldpc}
In general, the \emph{syndrome extraction circuit} used to measure the stabilizer of the code is subject to noise. 
As such, syndrome errors need to be accounted for in fault-tolerant QEC protocols. 
Designing fault-tolerant syndrome extraction circuits can add considerable overhead.
For instance, the Shor fault tolerant scheme requires at least $\propto d^3$ check measurements~\cite{shor_fault-tolerant_1996, delfosse_beyond_2022}.
Dennis~\etal \cite{dennis_topological_2002} argue that for a two-dimensional surface code with distance $d$, $O(d)$ rounds of noisy syndrome measurements need to be conducted to achieve fault-tolerance.

As an alternative to time-domain decoding, Bombin~\cite{bombin_single-shot_2015} showed that there exist \emph{single-shot codes} for which a single round of noisy syndrome measurements suffices.
One of the main advantages of single-shot codes is that the complexity of decoding is reduced since the structure of the decoding problem does not have an additional time dimension---as is the case for non-single-shot codes. 
Moreover, the time needed to conduct a QEC cycle is reduced since only a single set of stabilizer measurements needs to be done, which results in the fact that more logical operations can be conducted in a time interval compared to time-domain decoding.
The explicit construction of single-shot codes has been explored recently~\cite{campbell_theory_2019, quintavalle_single-shot_2021}, where the central idea is to use redundancy in the checks to ensure the single-shot property.
In Ref.~\cite{quintavalle_single-shot_2021}, 
it is shown that single-shot codes with necessary redundancy can be constructed using tensor products of chain complexes (i.e., three-dimensional hypergraph product constructions).
Intuitively, in this construction, the extra dimension yields an additional (classical) code that can be used to detect syndrome errors.
In this case, we can define an additional set of classical checks called \emph{metachecks} with check matrix $M$ that defines the corresponding classical linear \emph{metacode} $\calM{}$. 
The metacode satisfies the condition $M\cdot H_{X/Z}=0$, which means that every syndrome that can originate from $H_{X/Z}$ is a codeword of the metacode $\calM$.
The metacode $\calM$ is used to determine whether a syndrome is valid, i.e., we can use the metachecks to compute a \emph{meta-syndrome} $s_m := M \cdot s_{X/Z}$, which can then be decoded to fix syndrome failures. 

Geometrically, in the 3DSC, we can associate metachecks with the volumes of the lattice, as depicted in~\sfigref{fig:3dsc}{c}. 
This gives a \qgate{X} metacode $\calM_X$ whose checks act on the syndromes induced by the \qgate{X}-face checks. 
Intuitively, the volume metachecks help to close noisy loop-like syndromes.

\subsection{Noise model}
In this section, we describe the noise model that is used for the investigation of the various decoding methods (cf.~\secref{sec:decoding-main}) as well as the quasi-single-shot protocol (cf.~\secref{sec:qss-codes}). More details about our simulation procedures can be found in~\appref{sec:numeric-sim-details}. 

We consider a biased phenomenological noise model with analog syndrome measurements inspired by stabilized cat qubits.
We assume a three-dimensional architecture in which bosonic cat qubits are used as the inner code and a \mbox{$[[n=3 L^3,k=3,d_X=L,d_Z=L^2]]$} three-dimensional surface code with periodic boundaries is used as the outer code. 
i.e.,  the $n$ physical qubits are cat qubits, 
so qubits encoded in continuous-variable quantum systems
with Hilbert space ${\cal L}(\R^2)$,
which are used to protect $k$ logical qubits with $X$ ($Z$) distance $d_X = L$ ($d_Z = L^2)$. That is to say, 
this noise model affects the description on two levels: the level of qubits in the outer code, and the level of continuous variable bosonic modes in the inner code.
Note that our phenomenological noise model, while capturing key features of cat qubit noise, does not explicitly correlate effective error rates to physical noise and confinement parameters.
Our main motivation for using this model is to retain generality and applicability to related inner qubit noise models while capturing the key features of cat qubit noise.

In each QEC cycle, the stabilizers are measured, yielding an analog syndrome as a real number for bosonic codes.
Intuitively, the magnitude of the analog measurement can be interpreted as information on the reliability of the syndrome readout and thus can be used in the overall decoding routine. 
Note that in this work we use a phenomenological noise model and hence do not consider a specific implementation of the syndrome extraction circuit. 
For example, the standard Steane scheme~\cite{steane_error_1996} could be used.

\subsubsection{Qubit error noise model}

Here, we outline specifics concerning the error model on the level of qubits in the outer code.
We consider a model defined as a quantum channel that takes the form of a diagonal Pauli channel, reflecting stochastic Pauli noise. Such a noise model can be obtained from a non-diagonal error model via a group twirl.
Let $\rho$ denote the density matrix of a single qubit state that undergoes the quantum noise. 
Then, the Pauli error-channel $\calE$ has the Kraus representation
\begin{equation}
    \calE{}(\rho) := (1-p) \rho + p_X X \rho X + p_Y Y \rho Y + p_Z Z \rho Z. \label{eq:noise-channel}
\end{equation}
The single qubit physical error rate $p_{\rm err}\in [0,1]$ is the total probability of \qgate{X},\qgate{Y}, and \qgate{Z}-type errors. 
The individual Pauli error rates $p_{\Omega},\, \Omega \in \{X, Y, Z \}$ are specified by the physical error rate $p_{\text{err}}$ and a bias vector $\vec{r} = (r_X, r_Y, r_Z)$. This bias $\eta_\Omega$ for a certain error species $\Omega \in \{X, Y, Z \}$ is given by 
\begin{equation}
    \label{eq:bias_eq}
    \eta_{\Omega} := \frac{r_{\Omega}}{\sum_{\Omega'\neq \Omega} r_{\Omega'}} \, p_{\text{err}}.
\end{equation}
We call, for example, a noise channel \qgate{Z}-\emph{biased}, if 
\begin{equation}\label{eq:z-biased-chnl}
\eta_Z = \frac{r_Z}{r_X + r_Y} > 0.5.    
\end{equation}

\subsubsection{Syndrome error noise model}\label{sec:syndr-noise-model}
Inspired by the physical realization of bosonic cat qubits, we model the syndrome (or check) errors as Gaussian random noise that is added to the continuous syndrome $x\in\mathbb{R}$ obtained in the measurement. This affects the readout of the 
noiseless stabilizer values $s_i = \pm 1$. 
Therefore, the noisy stabilizer values $\tilde{s}_i$ have a continuous outcome given by $\tilde{s}_i = s_i + x_i$ where $x_i \sim \mathcal{N}(0, \sigma_i^2)$ is a Gaussian random variable with mean 0 and variance $\sigma^2$.
By thresholding the noisy syndrome $\tilde{s}$, i.e., taking signs, $\text{sgn}(\tilde{s_i})$, we obtain the \emph{hard syndrome}. 
As detailed in \appref{sec:s-noise-model-conversion}, the thresholding procedure allows us to relate syndrome error rates $p_{\text{err}}^{\text{synd}}$ and the variance $\sigma^2$ of the Gaussian noise process. 
i.e.,  given $p_{\text{err}}^{\text{synd}}$, the associated standard deviation $\sigma$ is given by
\begin{align}\label{eq:perr-to-sigma}
    \sigma = \frac{1}{\sqrt{2} \text{Erfc}^{-1}(2 p_{\text{err}}^{\text{synd}}) },
\end{align}
where $x\mapsto \text{Erfc}^{-1}(x)$ is the inverse of the complementary error function.

When considering biased-noise error models, we also bias the syndrome error channel in the same way as we bias the qubit error channel discussed in the previous section.
From the individual syndrome error rates for \qgate{X} and \qgate{Z} checks, we obtain through Eq.~\eqref{eq:perr-to-sigma} the corresponding variance of the Gaussian random noise model.
For example, a large \qgate{X} bias means that there will be more bit-flip errors, as well as more \qgate{X}-syndrome errors compared to phase-flip errors and associated syndrome errors. In other words, 
we model the error affecting the bosonic modes of the inner code and in consequence the qubits of the outer code in a phenomenological fashion.

When incorporating the analog information into the decoding process, we replace the \emph{log-likelihood ratios} (LLRs) of a discrete error model $\gamma_{\rm synd} = \log[(1 - q) / q]$, where $q$ is the measurement error probability, with 
\begin{equation}
    \label{eq:llr_analog_information}
    \gamma_{i} = \log \left[ \frac{\mathrm{Pr}(\tilde{s}_i \vert s_i = +1)}{\mathrm{Pr}(\tilde{s}_i \vert s_i = -1)} \right] = \frac{2 \tilde{s}_i}{\sigma^2},
\end{equation}
where $\mathrm{Pr}(\tilde{s}_i \vert s_i = \pm 1)$ is the probability of observing the noisy syndrome $\tilde{s}_i$ under the condition that the noiseless syndrome value is $s_i = \pm 1$.

\section{Min-sum belief propagation algorithms}\label{sec:decoding-main}
In this section, we discuss the decoding of quantum codes using iterative decoding procedures based on minimum-sum belief propagation (BP) decoding. 
First in~\secref{sec:hard-msa}, we review standard min-sum BP decoding for hard syndromes (i.e., for discrete variable codes), which forms the basis for the discussed decoding procedures. 
In~\secref{sec:si-decoder}, we review recent work on the use of analog syndrome information to decode quantum codes~\cite{raveendran_soft_2022}, propose improvements to these techniques, and discuss caveats of the original method.
The results of the numerical simulation are presented in the following section,~\secref{sec:atd}, together with a comparison to the proposed decoding technique, analog Tanner graph decoding.

\subsection{Hard syndrome MSA decoding}\label{sec:hard-msa}

\emph{Belief propagation} (BP) is an iterative algorithm that is known to be efficient in decoding classical LDPC codes and has been adapted to quantum codes successfully in recent years. 
BP is a message-passing algorithm operating on the Tanner graph of the code (also called the factor graph in this context). 
The graph is considered a model that describes the factorization of the joint probability distribution of the error. 
Given the measured syndrome, the goal of BP is to (approximately) compute the marginal probabilities for each bit.
For an error $e$ and syndrome $s=H\cdot e$, BP finds an estimate $\varepsilon$ of the error $e$ that yields the same syndrome $s=H\cdot\varepsilon$. 
The estimate vector $\varepsilon$ is formed as $\varepsilon=(\varepsilon_1, \dots, \varepsilon_n)$ where 
\begin{equation}\varepsilon_i:=\text{argmax}_{e_i}[P(e_i|s)].
\end{equation}

There are several variations of the standard BP algorithm that differ in the way marginal probabilities are computed. 
In the following, we focus on min-sum BP, which has been argued to be easier to implement on hardware-near devices than other variants~\cite{raveendran_soft_2022}.
For more details on min-sum BP, we refer the reader to~\appref{sec:app_bp}. 

When applied to factor graphs that are trees (i.e., that do not contain loops), BP is known to be exact and will converge to a solution within a number of iterations less than the diameter of the graph.
However, in general, LDPC codes contain loops.
In this setting, BP computes approximate marginals and is not guaranteed to converge to a solution that satisfies the syndrome equation.
In cases where BP does not terminate, the decoding can be deferred to a post-processing routine.
Such post-processing routines are typically more computationally expensive but will ensure the decoder returns a solution satisfying the syndrome equation.
A commonly used post-processing method to improve the overall decoding performance of BP algorithms is \emph{ordered statistics decoding} (OSD). This was first introduced in Ref.~\cite{fossorier_soft-decision_1995} and has recently been adapted to QLDPC codes~\cite{panteleev_degenerate_2021, roffe_decoding_2020}. 

\begin{figure}[!t]
    \centering
    \includegraphics{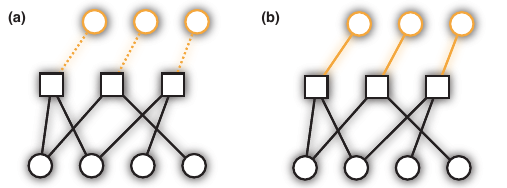}
    \caption{\textbf{(a)} Illustration of the tanner graph for SSMSA decoding.
    The analog information taken into account by SSMSA can be 
    illustrated similarly to ATD (although in SSMSA the analog information is \emph{not} directly incorporated in the factor graph), however, because of the cutoff parameter, it may happen that SSMSA discards the analog information, which corresponds to ignoring the virtual nodes, indicated by dashed edges.
    \textbf{(b)} Sketch of the \emph{analog Tanner graph} (ATG). 
    The subgraph colored in black corresponds to the Tanner graph of the code. 
    The subgraph highlighted in 
    yellow corresponds to the virtual nodes that are used to incorporate the analog information. 
    Their union is the ATG.}
    \label{fig:atg}
\end{figure}

Let us briefly summarize the main aspects of OSD. 
We consider a single check side in the following syndrome equation: 
\begin{equation}\label{eq:osd-syndr-eq}
H\cdot e = s. 
\end{equation}
In general, $H$ may not be a square matrix and may not have full rank, and therefore we cannot directly invert $H$ to solve Eq.~\eqref{eq:osd-syndr-eq}. 
The strategy applied is to choose a subset of columns that are linearly independent, and hence form a basis. 
Let $S:=\set{S_i}$ be the set of column indices of linearly independent columns we chose and $H_S$ be the matrix of column vectors obtained. Clearly, $H_S$ has full rank and can thus be inverted to give a solution 
\begin{equation}
    \label{eq:osd_invert_eq}
  H_S^{-1}\cdot s = e_S.  
\end{equation}
Different choices of $S$ correspond to (unique) solutions $e_S$. 
The main idea of BP+OSD is to use the marginal probabilities computed by the BP decoding algorithm to select a set $C$ that contains columns corresponding to bits with a high error probability.

\subsection{Soft-syndrome MSA decoder}\label{sec:si-decoder}
Recently, Ravendraan~\etal~\cite{raveendran_soft_2022} introduced a variant of iterative min-sum decoding called \emph{soft-syndrome min-sum algorithm} (SSMSA) to decode QLDPC codes using analog information.
We briefly review the main aspects of the algorithm, for which we present the first open-source implementation. 
We show that the results of 
our implementation match the original results presented in Ref.~\cite{raveendran_soft_2022}.
Furthermore, we explore the combination of SSMSA with \emph{ordered statistics decoding} (OSD), which improves the decoding performance compared to the original work. 

SSMSA is an iterative message-passing decoding procedure, i.e., a variant of (min-sum) \emph{belief propagation} (BP) that uses soft information, i.e., real syndrome values instead of hard (binary) ones in its update rules. Initially, the corresponding binary syndrome is obtained by ``thresholding'' the analog syndrome, i.e., determined by the signs of the analog values.
The update rules that are used to compute the messages are equivalent to those used in min-sum BP decoding with the addition of a ``cutoff'' parameter $\Gamma$, which is used to determine if the analog syndrome information should be considered in the update rules or not.
If the absolute value of the analog syndrome is below the cutoff, it is taken into account when computing the min-sum updates (in addition to the standard messages).
Conversely, if the absolute value of the analog syndrome is above the cutoff, the standard min-sum rules are applied. 
The SSMSA decoding process is visualized in~\sfigref{fig:atg}{a} and the detailed pseudocode is presented in~\Cref{alg:ssmsa} in~\appref{sec:app_si-decoder}.

Analogously to BP, it is not guaranteed that SSMSA converges. 
The algorithm tries to infer a decoding estimate based on marginal probabilities, i.e.,  it tries to find the most probable error given a syndrome in a given number of maximum steps.
Hence, in case the algorithm terminates due to reaching the maximum number of steps, one can in principle use the marginal probabilities the algorithm computed up to termination in a \emph{post-processing} step to infer a decoding estimate.
However, post-processing techniques were not considered in Ref.~\cite{raveendran_soft_2022}.

Note that in SSMSA, soft syndromes are not directly included in the parity-check matrix $H$ (i.e., the factor graph used for decoding is not altered), but are only used as an additional parameter to compute the marginal probabilities during the iterative decoding procedure.
This means that measurement errors will not be considered for possible fault locations that satisfy the syndrome in OSD post-processing, i.e., there are situations in which we are trying to solve for a syndrome that is not in the image of $H$.
In that case, it can happen that the ``solution'' with the highest error probability is not a solution of Eq.~\eqref{eq:osd_invert_eq}.
This leads to cases where OSD post-processing does not give a significant improvement over standard SSMSA decoding.
In the following section, we propose techniques that amend this problem by ensuring that the factor graph is always invertible, i.e., has full row rank.

\begin{figure*}
    \centering
    \includegraphics{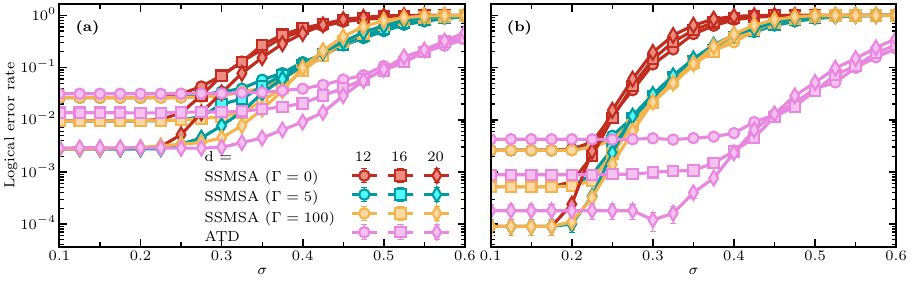}
    \caption{Comparison of the SSMSA decoder with ATD for increasing syndrome noise $\sigma$ on a family of lifted product codes with distances $d\in \set{12,16,20}$, cf.~\secref{sec:app_code_constructions}. The data error rate of the unbiased depolarizing noise model is fixed at $p=0.05$. \textbf{(a)} Comparison of the original SSMSA implementation and ATD using BP (without OSD). \textbf{(b)} Comparison of SSMSA+OSD and ATD using BP+OSD. The legend is shared between both panels and $\Gamma$ indicates the cutoff value of the soft information (SI) SSMSA decoder. Note that the case $\Gamma = 0$ corresponds to ignoring the analog information and, therefore, is equivalent to ordinary hard syndrome decoding.}
    \label{fig:decoder_comparisons}
\end{figure*}

\section{Analog Tanner graph decoding}\label{sec:atd}

In this section, we propose a novel decoding technique for quantum stabilizer codes using analog syndrome information that is based on min-sum BP+OSD decoding. 
We call this technique 
\emph{analog Tanner graph decoding} (ATD) as it is based on the construction of a variant of the standard Tanner graph that incorporates analog syndrome information.

In~\secref{sec:atd-ss}, we briefly review single-shot decoding of quantum codes. 
Then, in~\secref{sec:atd-ss}, we show how the analog Tanner graph decoding can be adapted to be used to decode single-shot quantum codes.
The numerical results in this section are obtained by standard Monte Carlo decoding simulations, which are discussed in more detail in~\secref{sec:numeric-sim-details}. 
Additionally, in \appref{sec:bias-pres-syndromes-cats} we discuss how the proposed techniques can be adjusted to also include analog information of the data qubits and not only of the syndrome readout in similar architectures as the one we focus on in the rest of this article.

Given the Tanner graph of a code $\calT = (V_C \cup V_D, E)$, we can directly incorporate analog information by adding $|V_C|$ additional data nodes, called \emph{virtual (data) nodes}, $V_V$. 
Each check node $V_C$, is connected to a single virtual data node through a single edge, as visualized in \sfigref{fig:atg}{b}. The virtual data node stores the probability that the check is in error. 
In our protocol, this probability is derived from the magnitude of the analog syndrome readout. 
We refer to the modified Tanner graph as an \textit{analog Tanner graph} (ATG).

In terms of parity check matrices, building the ATG amounts to appending an $m\times m$ identity matrix $\Id_m$ to the original parity-check matrix of the code: 

\begin{definition}[Analog check matrix]\label{def:atg}
    Given a parity-check matrix $H \in \F_2^{m\times n}$ corresponding to the incidence matrix of a Tanner graph $\calT(H)$, the analog check matrix $H^{\calA}$ has the following form
    \begin{equation}\label{eq:atg}
        H^\calA := [H \mid \Id_m].
    \end{equation}
\end{definition}

Clearly, $\calT^{\calA}$ constitutes a valid Tanner graph whose incidence matrix is $H^\calA$. Moreover, $H^{\calA}$ always has full rank.  Consequently, there always exists a solution to the syndrome equation, cf. Eq.~\eqref{eq:osd-syndr-eq}.
For decoding, the virtual data nodes are initialized with the analog syndrome LLRs given in Eq.~\eqref{eq:llr_analog_information}.
Note that this construction can be seen as a generalization of \emph{data syndrome codes} for standard, hard (binary) syndromes~\cite{kuo_decoding_2021, li_numerical_2020, grospellier_combining_2021, breuckmann_single-shot_2022, ashikhmin_robust_2014}, tailored to error models inspired by bosonic codes and their associated analog syndrome information.

To decode using the ATG, we can now apply standard decoding approaches such as BP+OSD, making the approach flexible to different decoding strategies and variations of the algorithm. 
If not stated otherwise, we refer to decoding using the ATG with BP+OSD as \emph{analog Tanner graph decoding} (ATD).
Note that the explicit inclusion of virtual nodes that correspond to fault locations for measurement errors eliminates the problems encountered with SSMSA and the OSD post-processing techniques mentioned at the end of~\secref{sec:si-decoder}.
Moreover, even though SSMSA with a cutoff of $\Gamma=\infty$ is conceptually similar to ATD, the syndrome update rules in SSMSA (as discussed in the previous section) and the fact that the information is not directly incorporated into the factor graph used for decoding, leads to significant performance discrepancies.

To investigate the decoding performance and compare ATD to SSMSA(+OSD), we perform standard Monte Carlo simulations to estimate the logical error rate for increasing syndrome noise $\sigma$ (and fixed data error rate $p=0.05$) on a family of \emph{lifted product} (LP) codes. 
The codes are taken from a code family defined in Ref.~\cite{raveendran_finite_2022}.
For more detail on the construction of codes via the lifted product, we refer the reader to~\secref{sec:app_code_constructions} in the appendix.

In~\sfigref{fig:decoder_comparisons}{a}, we compare the performance of SSMSA and ATD in the absence of OSD post-processing, i.e., only BP is used for ATD (ATD+BP) and the original SSMSA algorithm is used.
Following the methodology of Ref.~\cite{raveendran_soft_2022}, we fix the data error rate of the unbiased depolarizing noise model to be $p = 0.05$ and vary the syndrome noise. The strength of the syndrome noise is characterized by the standard deviation $\sigma$ of the Gaussian noise channel as defined in~\secref{sec:syndr-noise-model}.
From~\sfigref{fig:decoder_comparisons}{a}, it can be seen that ATD+BP always outperforms SSMSA if the logical error rate is limited by the data noise.
We examined different cutoff values for the SSMSA implementation and found that a large value of $\Gamma$ yields the best performance for the code and noise parameter settings considered, but that there is no significant performance difference between $\Gamma = 100$ and values $\Gamma > 100$.
Additional discussions on different decoder parameterizations are presented in~\appref{app:additional_results}.
In particular, setting the cutoff for SSMSA to $\Gamma=0$ (meaning that the decoder does not take the analog syndrome information into account) clearly leads to higher logical sub-threshold error rates.

In \sfigref{fig:decoder_comparisons}{b}, the results of the same simulation setup as in \sfigref{fig:decoder_comparisons}{a} but using OSD post-processing for SSMSA and ATD are shown.
We observe that OSD always improves the performance of ATD. However, OSD leads to a reduced threshold when combined with SSMSA.
Furthermore, we observe higher logical error rates for SSMSA-OSD when $\sigma > 0.35$, as well as the value of the cutoff $\Gamma$ becoming less relevant to the decoder performance.
This illustrates the possible issues that can occur when combining SSMSA and post-processing without further modifications as discussed in the previous subsection.

\subsection{Comparison with SSMSA}

The ATD method and SSMSA are conceptually similar, but differ in their implementations. 
First, we focus on the SSMSA cutoff parameter $\Gamma$, an input parameter that the decoding performance explicitly depends on.
The two extreme values of $\Gamma$ correspond to $\Gamma = 0$ and $\Gamma = \infty$.
For the former case, the syndrome is ``trusted'' and the analog information is never taken into account. 
In this scenario, the virtual updates are ignored.
In the case where $\Gamma = \infty$, the analog information is always taken into account. 
The performance of the SSMSA algorithm is therefore related to the choice of $\Gamma$. 
In contrast, ATD does not have a cutoff parameter: the soft information is incorporated directly into the decoding graph and the amount of ``trust'' assigned to the syndrome value is determined through the standard MSA update rules.

Secondly, it is nontrivial to adapt SSMSA to more general Tanner graph constructions that go beyond the case of single-shot decoding. 
For example, SSMSA does not extend directly to the case of time-domain decoding or single-shot decoding with meta-checks~\cite{quintavalle_single-shot_2021, higgott_improved_2023}. By comparison, the ATD decoder works out of the box in both of these settings as it considers a different factor graph for decoding.

Finally, when the magnitude of the analog information is below the cutoff, $\Gamma$, the SSMSA update rules can overwrite the original analog syndrome information, (cf. Line~\ref{line:ssmsa-overwrite} in~\Cref{alg:ssmsa}). 
The remaining SSMSA decoding iterations then no longer have access to the initial anolog syndrome information.
In contrast, this scenario does not occur in ATD, because the analog syndrome information is stored explicitly in the nodes of the factor graph and the initial values are used throughout the algorithm by definition of the update rules.
Hence, the initial analog syndrome information remains accessible in all ATD decoding iterations. 
As such, we argue that ATD is better positioned to make full use of the anolog syndrome information as the information propagates through the decoding graph explicitly.

\subsection{Single-shot decoding with metachecks}\label{sec:ss}

The three-dimensional surface code is defined by three parity-check matrices: $H_X$, $H_Z$, and $M_X$. 
The first two are the standard CSS code matrices that define the \qgate{X}- and \qgate{Z}- stabilizers. 
The matrix $M_X$ defines a so-called \textit{metacode} $\calM_X$, which is designed to provide protection against noise in the \qgate{Z}-syndromes (i.e., \qgate{X}-checks). 
More precisely, the metacode is defined so that all valid (non-errored) \qgate{Z}-syndromes are in its code space. i.e.,  $M_X \cdot s_Z = 0$, when $s_Z=H_X\cdot e_Z$ for all $e_Z \in \mathbb{F}_2^n$.

In Ref.~\cite{quintavalle_single-shot_2021}, Quintavalle~\etal defined the \textit{two-stage} single-shot decoder for the three-dimensional surface codes and related homological product codes. 
The steps of the two-stage single-shot decoding protocol are as follows:

\begin{itemize}
    \item \textbf{Stage 1, syndrome repair:} measure the noisy syndrome $s=s_Z + s_e$, where $s_Z$ is the perfect syndrome and $s_e$ is the syndrome error. 
    Solve the metacode decoding problem: given the metasyndrome \mbox{$s_M=M_X \cdot s$}, find an estimate $s'$ for the noiseless syndrome.
    If the metadecoding is successful, we obtain a corrected syndrome $s_C=s+s'$ satisfying $M_X\cdot s_C = 0$. 
    \item \textbf{Stage 2, main code decoding:} use $s_C$ to obtain the decoding estimate $e'_Z$ by solving $s_C= H_X\cdot e_Z$. 
\end{itemize}
Quintavalle~\etal demonstrated that three-dimensional surface codes can be decoded using an implementation of the two-stage protocol where the first stage uses a \emph{minimum-weight perfect matching} (MWPM) decoder and the second stage a BP+OSD decoder.

A problem with the two-stage single-shot decoder is that the syndrome repair stage is independent of the main decoding stage. 
As a result, syndrome repair is always prioritized over applying corrections to the data qubits, even in situations where it would be more efficient to apply a combined correction. 
Furthermore, in certain circumstances, the two-stage decoder is subject to a failure mode whereby the syndrome repair step results in an invalid ``corrected'' syndrome $s_C$ that is not in the image of the parity-check matrix $H_\qgate{X}$, s.t. $s_C \notin \Im{H_\qgate{X}}$. 
Quintavalle~\etal proposed a subroutine to handle such failure modes, but this adds to the computational run-time of the protocol \cite{quintavalle_single-shot_2021}.

Recently, Higgott and Breuckmann~\cite{higgott_improved_2023} proposed a \emph{single-stage} decoding protocol for single-shot QLDPC codes that solves the problems mentioned earlier with the two-stage decoder. 
In the following we drop subscripts for simplicity and let $H$ be a check matrix, and $M$ be the corresponding metacheck matrix.
The single-stage approach considers the decoding problem of given a noisy syndrome $s = He+s_e$ and meta syndrome $s_M = Ms$, to find a minimum-weight estimate $(e', s'_e)$, such that $H^{M} \cdot (e',s'_e)^T = (s, s_M)^T$, where
\begin{equation}
    \label{eq:single_stage_meta_check_code}
    H^{M}:= \left(\begin{matrix}
        H & \Id_m \\
        0 & M
    \end{matrix}\right),
\end{equation}
is called the \emph{single-stage parity-check matrix}.
The single-stage parity matrix combines both the decoding of the data errors $e$ and the syndrome errors $s_e$. This ensures that the decoder uses the full information available to it, in contrast to the two-stage decoder that processes the meta-syndrome and syndrome separately. Furthermore, note that the first block-row of Eq.~(\ref{eq:single_stage_meta_check_code}) is full rank. This property guarantees that all solutions are valid, meaning that the single-stage decoder does not suffer from the failure mode that arises in the two-stage decoding approach.

\subsection{Analog single-shot decoding}\label{sec:atd-ss}

We now discuss how ATD can be applied to improve single-stage single-shot decoding and explore connections between single-stage parity-check matrices and the ATG construction. 
To begin, we first note that the single-stage parity-check matrix of Eq.~\eqref{eq:single_stage_meta_check_code} is conceptually similar to the analog parity-check matrix of Eq.~\eqref{eq:atg}: the single-stage parity-check matrix simply introduces a particular set of constraints described by the metacode $M_X$. 
Note that, however, as these constraints are linear combinations of the top full-rank block $(H \mid \Id_m)$ of $H^M$. 
Hence $H^M$ can be understood as making the implicit meta constraints explicit in the factor graph~\cite{higgott_improved_2023-1}.
\begin{figure}[b]
    \centering
    \includegraphics{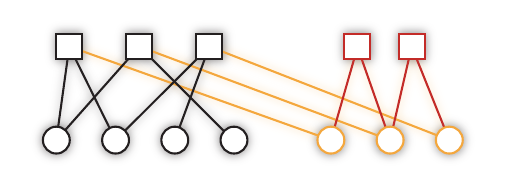}
    \caption{Sketch of the analog single-stage Tanner graph based on $H^{M}$.
    The black checks, bits, and edges correspond to the Tanner graph of the linear code $H$.
    The orange nodes are due to the identity $\Id_m$ in Eq.~\eqref{eq:single_stage_meta_check_code} and are initialized with the analog information of the syndrome readout. 
    The red nodes are due to the metacode, i.e., the metacheck matrix $M$.
    }
    \label{fig:ss-tg}
\end{figure}
An example of a single-stage Tanner graph is depicted in~\figref{fig:ss-tg}. 
From this sketch, it is evident that this construction corresponds exactly to the structure of the ATG together with additional (meta) check nodes as defined by the metacode $\calM_X$.
If we incorporate the metacode directly in the construction of the ATG, (similar to Ref.~\cite{higgott_improved_2023}), we construct the \emph{single-stage ATG}, which can be used to decode single-shot codes using analog information.

Let us briefly review some details of the overall decoding procedure.
Consider an analog single-stage data syndrome $s\in \R^{r_z}$. 
To initialize the ATD, we threshold $s$ to a binary syndrome $s_b\in \set{0,1}^{r_z}$, which we use to decode. 
Additionally, we use the analog syndrome values of $s$ to initialize the virtual nodes of the ATG, as in standard ATD. 
Standard decoding algorithms such as BP+OSD can then be applied to the resulting factor graph to obtain a decoding estimate. 

\begin{figure*}[tbh]
    \centering\includegraphics{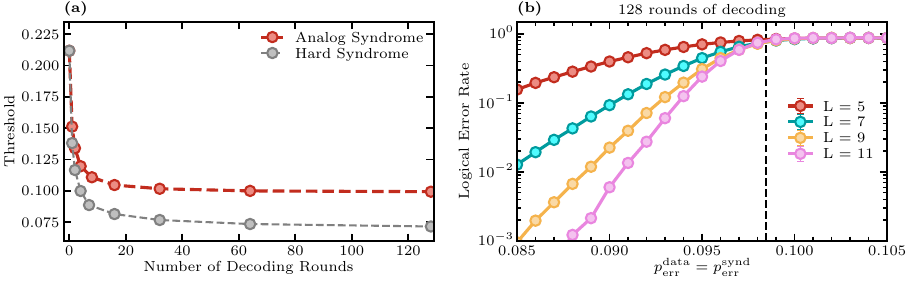}   
    \caption{Performance of single-shot analog Tanner graph decoding (ATD) for 3D toric codes.
    \textbf{(a)} Sustainable threshold of the single-shot side of 3D toric codes using ATD. 
    The analog information increases the sustainable threshold of current state-of-the-art methods~\cite{higgott_improved_2023} by almost $3 \%$.
    The results labeled ``hard syndrome'' are due to Higgott and Breuckmann~\cite{higgott_improved_2023}, while the ``analog syndrome'' results are obtained with our ATD method.
    \textbf{(b)} Example of threshold determination after 128 rounds of decoding using toric codes with lattice size $L \in \{5, 7, 9, 11\}$.
    }
    \label{fig:sus-th-atd-oscar-comparison}
\end{figure*}

In order to investigate the decoding performance of single-stage single-shot ATD decoding, we conduct sustainable threshold simulations for a concatenated bosonic 3D toric code using a phenomenological noise model. 
Single-shot error correction will generally leave some residual error after the correction. 
The goal is to suppress the accumulation of this residual error to the point at which it can be corrected in subsequent rounds.
To this end, we define the \emph{sustainable threshold} for a single-shot code as the physical error rate, $p_{\rm sus-th}$, below which the residual error remains constant and the quantum information can be stored indefinitely by increasing the distance of the code. 
More precisely, the sustainable threshold is defined as
\begin{equation}
    p_{\rm {sus-th}} := \lim_{R \to \infty} p_{\rm th}(R),
\end{equation}
where $p_{\rm th}(R)$ is the threshold for $R$ rounds of noisy stabilizer measurements. 

To estimate $p_{\rm sus-th}$ numerically for the 3D toric code, we estimate $p_{\rm th}(R)$ for increasing values of $R \in \set{0, 1, 2,4,8,\dots, 128}$ until the value $p_{\rm th}(R)$ is constant, i.e., until the threshold does not decrease when increasing the number of rounds $R$.
Our simulation results are shown in \sfigref{fig:sus-th-atd-oscar-comparison}{a}. 
We see that under single-stage ATD decoding, the 3D toric code has a sustained threshold of $9.9\%$. This improves by $2.8\%$ on the sustained threshold of $~7.1\%$ obtained by Higgott and Breuckmann for the same family of 3D toric codes, but using DV qubits and hard syndromes  \cite{higgott_improved_2023}. 
The improved sustainable threshold we observe highlights the benefits of considering analog information in decoding. ~\sfigref{fig:sus-th-atd-oscar-comparison}{b} shows the threshold of the $9.9\%$ for concatenated bosonic 3D toric codes of size $L=\{5,7,9,11\}$ after $128$ noisy rounds of decoding.

We note that the sustainable thresholds we have found could be further optimized by fine-tuning the parameters of the BP+OSD decoder.
However, rather than showing optimal improvements, the goal of this experiment is to indicate the improvements that become possible by considering analog readout information.
A more complete study would additionally compare below-threshold error rates, which are on this scale more relevant than the improvement of the threshold.
We will leave this question open for future work, including a more detailed simulation including circuit-level noise models.

Finally, we note that the check matrix defined in Eq.~\eqref{eq:single_stage_meta_check_code} cannot be used in the SSMSA algorithm without significant modifications.
The reason for this is that syndrome errors are already included explicitly in Eq.~\eqref{eq:single_stage_meta_check_code}, which can a priori not be handled by SSMSA, since for SSMSA the analog information is an input parameter and the algorithm operates on the standard parity-check matrix of the code.
Moreover, the metachecks do not correspond to physical measurements, i.e.,  they do not have analog syndrome information.

\section{Quasi-single shot codes}\label{sec:qss-codes}
Even in the presence of strong noise bias, as in the case of cat-LDPC architectures, both error species, \qgate{X} and \qgate{Z} noise, must be corrected to achieve good overall logical fidelity~\cite{chamberland_building_2022}. 
When implementing one-sided single-shot codes, such as the three-dimensional surface code, we cannot rely solely on the single-shot side of the code to correct errors.

To decode a quantum code that does not have the single-shot property, multiple rounds of (noisy) syndrome measurements must be performed so that the decoder infers the presence of measurement noise ~\cite{dennis_topological_2002, delfosse_beyond_2022, shor_fault-tolerant_1996, fowler_high-threshold_2009}.
Usually, this process is referred to as \emph{repeated measurements} or \emph{time-domain decoding}, since repeating stabilizer measurements adds an additional time dimension when considering the decoding instance. 
In this section, we investigate the decoding of quantum codes under phenomenological noise with repeated measurements in the presence of analog information. 
To this end, in~\secref{sec:atd-ow}, we discuss \emph{overlapping window} decoding~\cite{dennis_topological_2002} that we
generalize to decode QLDPC codes over time in the presence of analog information. 
Moreover, we elaborate on the relation between the overlapping window method and the ATG construction proposed in~\secref{sec:atd}. 

Motivated by the structure of 3D bosonic-LDPC code architectures, we further propose a novel decoding protocol that we call \emph{$w$-quasi single-shot codes} ($w$-QSS codes). 
The main idea is to leverage noise bias and analog syndrome information (provided by bosonic-LDPC codes) to demonstrate that only a small number $w \ll d$ of repeated syndrome measurement cycles suffices for the non-single-shot check side to give an overall decoding protocol with high logical fidelity for error rates sufficiently below threshold. 

The central result of this section is that we demonstrate numerically that for the cat-3DSC with reasonable code sizes ($L=11$), the $w$-QSS protocol achieves a threshold of $\approx 1.5\%$ for the non-single-shot side under phenomenological noise and that in the sub-threshold regime, $w=3$ suffices to match the decoding performance of time-domain decoding with $\propto d$ repeated measurements.

\subsection{Analog overlapping window decoding for QLDPC codes}\label{sec:atd-ow}
To decode an $[[n,k,d]]$-quantum code under phenomenological noise, i.e., when the syndrome measurements are noisy, we need to repeat the measurements several times (usually at least $d$ times) in order to handle the noisy syndromes~\cite{dennis_topological_2002, delfosse_beyond_2022}. 
Generally, this leads to the dimension of the decoding instance being increased by one, as the time dimension is now also being considered. Hence, we refer to this decoding problem as \emph{time-domain decoding}.
For example, the decoding problem of a repetition code under phenomenological noise leads to a two-dimensional decoding problem, analogous to the two-dimensional surface code under bit-flip noise.
Similarly, the decoding of the two-dimensional surface code over time leads to a three-dimensional problem, analogous to decoding the three-dimensional surface codes 
under bit-flip noise. 
We make the connection between the proposed construction of ATG, the decoding of quantum codes over time, and the tensor product chain complexes explicitly and formally argue that they are equivalent (cf.~\Cref{prop:informal}).

\begin{figure}[!b]
    \centering
    \includegraphics{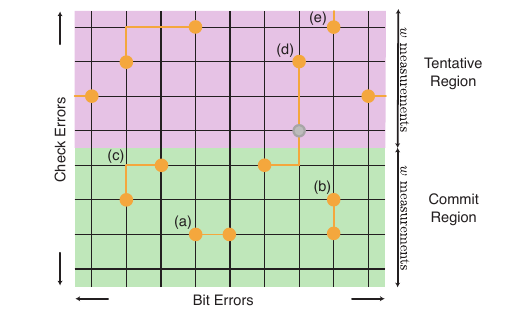}
    \caption{Sketch of the overlapping window decoding method of a repetition code as proposed in Ref.~\cite{dennis_topological_2002}. 
    \textbf{(a)} A space-like single data qubit error, \textbf{(b)} A time-like error, \textbf{(c)} A space-time error, \textbf{(d)} A space-time error that extends across the region boundary, \textbf{(e)} A time-like error that extends into the next decoding round.
    Note that the error \textbf{(d)} will only be partially corrected as it extends over the boundary of the commit region. 
    The inferred correction will imply that all defects in the commit region are removed, but will introduce a new defect in the tentative region shown as a gray dot.
    Defects created in this way are referred to as \emph{virtual defects}~\cite{skoric_parallel_2022} in the literature.
    An error purely residing in the tentative region will be considered during decoding to infer a decoding estimate matching the syndrome, but will not be corrected in the same round.
    }
    \label{fig:overlapping-window-repcode}
\end{figure}

\emph{Overlapping window decoding} (OWD) was originally introduced by Dennis~\etal in Ref.~\cite{dennis_topological_2002} to decode quantum surface codes over (finite) time. 
In OWD we divide the collected syndrome history into $w$-sized regions, the first two of which are sketched in~\figref{fig:overlapping-window-repcode}.
At any instance, the decoder computes the correction for two regions, whereas the older one is called ``commit'' and the newer one ``tentative'', $R_c, R_\tau$, respectively.
A window encompasses a total of $w$ noisy syndrome measurements. 
For each decoding round, the syndrome data of $2w$ rounds, i.e., two regions, is used to find a recovery operation by applying a decoder.
However, only the correction for the first $w$ rounds, i.e.,  for region $R_c$ is applied (by projecting onto the final time step of $R_c$ and applying the recovery accordingly). 
After applying the recovery, the region $R_c$ can be discarded, and only $R_\tau$ is kept. 
Then, in the next decoding round, the same procedure is repeated using the previous $R_\tau$ as the new commit region and the next $w$ rounds as the new temporary region, which is conceptually equivalent to ``sliding up'' the $2w$-sized decoding window one step.

Dennis~\etal argue for the two-dimensional surface code that the number of time steps $w$ should be chosen proportionally to the distance of the code, $w \gg d$, to ensure that the probability of introducing a logical operator is kept small.
This method is needed to simulate memory experiments over (a finite amount of) time, since the last, perfect round of measurements (corresponding to data qubit readout) may artificially increase the observed threshold leading to the fact that doing fewer repetitions always performs better. This aspect has also been observed in Ref.~\cite{pattison_improved_2021}.
Additionally, overlapping window decoding also corresponds to how a fault-tolerant quantum computer will likely be decoded in practice~\cite{skoric_parallel_2022}.

Let us at this point fix some terminology: the window size $w$ is the number of syndrome measurement records in a single window, i.e., the total size of the two regions is given by $|R_c| + |R_\tau| = 2 w$.
A single round of decoding takes $2 w$ noisy syndrome measurements into account, however, due to the sliding nature of the decoding window, $n$ rounds of decoding correspond to $(n+1) w$ syndrome measurements.
We refer to time-domain decoding, where the number of repetitions is proportional to the distance of the code as \emph{standard decoding}.
To apply this method to arbitrary QLDPC codes, we first propose the construction of the \emph{3D analog Tanner graph}, i.e., the decoding graph over time for time-domain decoding. 

\begin{figure}[t]
    \centering
    \includegraphics{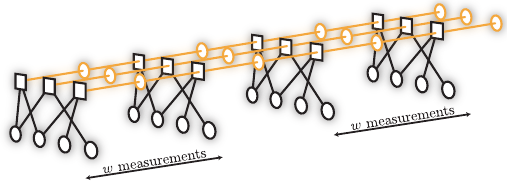}
    \caption{Overlapping window decoding for QLDPC codes (arbitrary Tanner graphs).
    The individual Tanner graphs of a code described by the parity-check matrix $H$ are ``glued'' together by the orange syndrome nodes that correspond to time-like (measurement) errors.
    A single decoding round with window size $w$ involves $2 w$ measurements, but only corrections in the first $w$ are committed, see also \figref{fig:overlapping-window-repcode}. }
    \label{fig:overlapping-window-ldpc}
\end{figure}

Given the Tanner graph of a QLDPC code $\calT(\calC)$, we first create copies of $\calT$ and introduce an additional set of bit nodes between pairs of checks between consecutive copies of $\cal{T}$, and an additional set of bit nodes for the last copy of $\calT$. 
These correspond exactly to time-like errors on syndrome nodes. The construction is sketched in~\figref{fig:overlapping-window-ldpc}.
Algebraically, the multi-round parity-check matrix $\tilde{H}$, i.e., the incidence matrix of the multi-round Tanner graph can be defined as follows.
\begin{definition}[Multi-round 
parity-check matrix and Tanner graph]\label{def:mr-pcm}
    Given an $m\times n$ parity-check matrix $H$, the $r-$multi-round parity-check matrix is defined as
    \begin{align}\label{eq:multi-round-pcm}
        \tilde{H} & := 
            \left(
                \begin{matrix}
                    H_{\rm diag} \mid \Id_{\rm sdiag}
                \end{matrix}
            \right) \\ 
            &=
        \left(
            \begin{matrix}
                    H & 0 & 0 & \dots & \Id_{m} & \\
                     & H & 0 & \dots & \Id_{m} & \Id_{m} & \\
                     &  & H & 0 & \dots & \Id_{m} & \Id_{m} & \\
                    & & &\ddots  & & \ddots \\
                     &  &  & H &  0 & 0 & & \Id_{m}  & \Id_{m}
            \end{matrix} 
        \right),
    \end{align}
    where $H_{\rm diag}$ is a block-diagonal matrix with $r$ diagonal block entries, and $\Id_{\rm sdiag}$ has a step-diagonal form, consisting of $m\times m$ identity matrices.
\end{definition}

Hence, $\tilde{H}$ is an LDPC code whose parity-check matrix consists of copies of $H$ with additional bit nodes between pairs of checks. 
It is easy to see that this code is LDPC (w.r.t.~the number of repetitions) since the vertex degree for each check is increased by (at most) 2, independent of the number of repetitions.
As an example, consider the multi-round (analog) Tanner graph and the corresponding multi-round check matrix depicted in~\figref{fig:fig1}{f}, which corresponds to an instance of a multi-round parity-check matrix (and Tanner graph) for a repetition code over two rounds. 

We show that the construction of the multi-round Tanner graph for $r$ rounds of syndrome measurement can be described as the tensor product chain complex of the chain complex corresponding to the QLDPC code $\calC{}$ and the chain complex of (a slight variant of) the $r$-repetition code $\calR$:

\begin{proposition}[Informal statement]\label{prop:informal}
    The $r$-multi-round parity-check matrix $\tilde{H}$ of $\calC$ is equivalent to the check matrix of the tensor product code $\mathcal{R} \otimes \calC$.
\end{proposition}

The proof is elementary and follows from the chain complex tensor product; we refer the reader to \appref{ssec:proofs} for more details.
Note that this highlights that the construction is equivalent to ``stacking'' the ATG constructed from $H$ and connecting the virtual nodes between pairs of checks.
This gives a straightforward correspondence between ATGs and phenomenological check-matrices and allows us to directly apply ATD to $\tilde{H}$.

The overall procedure for multi-round decoding with analog information can be summarized as follows.
First, from the parity-check matrix $H$ build the $r$-multi-round check matrix $\tilde{H}$ (corresponding to the multi-round analog Tanner graph).
Second, use the virtual nodes (corresponding to time-like data nodes in standard models) to incorporate analog syndrome information.
And finally, apply the ATD decoder on $\tilde{H}$.

Note that recently and independently, BP+OSD has been used to decode (single-shot) LDPC codes under a circuit-level noise model~\cite{bravyi_high-threshold_2023, xu_constant-overhead_2023}.
The proposed overlapping window approach we outline in this work differs in that it also considers analog information in the decoding. Additionally, our methods apply to QLDPC codes in general and do not rely upon codes with a special code structure.
Furthermore, recently and independently in Ref.~\cite{kuo_correcting_2023}, the authors considered decoding of LDPC codes under discrete phenomenological noise using a check-matrix construction equivalent to our definition of the multi-round parity-check matrix. 
However, they do not employ overlapping window decoding.

\subsection{Quasi single-shot decoding}\label{sec:qss-decoding}

In the previous section, we discussed how ATD can be used together with multi-round (analog) parity-check matrices to decode under phenomenological noise with analog information. 
In this section, we propose a novel protocol based on these techniques that lowers the overhead induced by repeated measurements. 

In the $w$-QSS protocol, we assume a QLDPC code where one side is single-shot and the other side is not---inducing the need for at least $d$ repeated measurements (or a number of repeated measurements proportional to the distance of the code) in the presence of noisy syndromes, where $d$ is the distance of the non-single-shot side. This applies, for instance, to three-dimensional surface codes, which we use as representatives in the following. 
The single-shot side of the code can be decoded using analog single-stage decoding discussed in~\secref{sec:atd-ss}.
Complementarily, a slight generalization of the analog \emph{overlapping window decoding}  (OWD) as described in~\secref{sec:atd-ow} is used to decode the non-single-shot side of the code.
The overall protocol is straightforward:

\begin{enumerate}
    \item Choose a $w \ll d$. 
    Intuitively, $w$ controls the number of noisy syndrome measurements to conduct for the non-single-shot side. 
    i.e.,  $w$ is the window size for overlapping window decoding.
    \item On the single-shot side of the code, we apply the usual analog single-shot decoding procedure, where in each time step we do a syndrome measurement and infer a recovery operation.
    \item For the non-single-shot side, we do $w$ time steps of repeated measurements and then decode (using analog OWD). 
\end{enumerate}

Without analog information, by restricting the number of syndrome measurements to $w$ the \emph{effective distance} of the code is lowered to $w$. 
For instance, consider a code that has $d_X=4, d_Z=16$ where the \qgate{X}-side is single-shot. 
Then, setting $w=4$ effectively reduces $d_Z$ to $4$ because, in the time dimension, logical errors can be of weight $4$ only. 
This is apparent when considering the multi-round Tanner graph used for decoding over time as a tensor product with a repetition code. 
If we have $d_Z=16$ and do $16$ rounds of noisy syndrome measurements, the repetition code protecting the system from ``timelike'' errors has distance $16$, but when restricting to $w$ repetitions, we essentially cut the repetition code in time to a shorter version, thereby lowering the distance along the time dimension.

To investigate the performance of the proposed protocol numerically, we simulate three-dimensional surface codes under phenomenological (cat) noise for lattice sizes $L = 5$ to $L = 11$ and compare different choices of $w$ versus the standard approach of taking a number of repeated measurements proportional to the distance. 
Note that, to conduct numerical simulations for repeated measurements we need to conduct multiple rounds of decoding to avoid overestimation of the threshold~\cite{pattison_improved_2021} (cf.~\secref{sec:atd-ow}). 
Since the non-single-shot side of this code can be decoded with matching-based algorithms, we use
PyMatching~\cite{higgott_pymatching_2022, higgott_sparse_2023} for decoding. 
The threshold behavior and the sub-threshold scaling for 32 decoding rounds are shown in~\figref{fig:qss-3dtoric}.
We discuss the decay of logical fidelity later on in \figref{fig:qss_logical_decay}. 
The main findings are as follows.
\begin{itemize}
    \item For $w=2$ we observe an increase in logical error rate for $L=11$ for error rates around the threshold, however, the sub-threshold suppression for lower error rates performs as for the standard time-domain decoding as shown in~\sfigref{fig:qss-3dtoric}{b}.
    \item $w=3$ is enough to match the results of standard time-domain decoding (for the investigated code sizes and error rates).
    \item For $w=1$ the protocol does not work.
\end{itemize}

\begin{figure*}[!]
    \centering
    \includegraphics{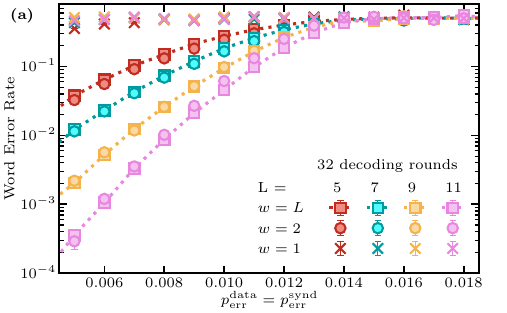}
    \includegraphics{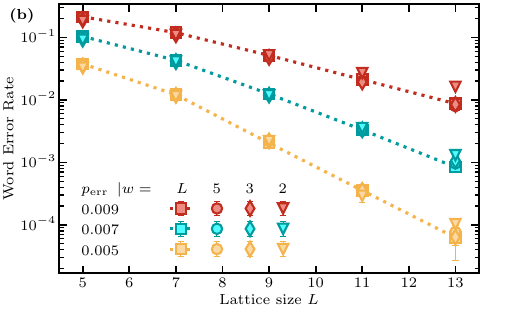}
    \caption{Performance of the $w$-quasi-single-shot protocol using the three-dimensional surface codes under phenomenological bit-flip noise with analog syndrome readout. The non-single-shot side is decoded using minimum-weight perfect matching.
    \textbf{(a)} Word error rate after 32 decoding rounds as a function of the phenomenological error rate for various window size values $w$. 
    \textbf{(b)} Below threshold scaling of the word error rate after 32 decoding rounds for various window size values $w$.
    The results suggest that it is sufficient to repeat the stabilizer measurement a finite number of times \emph{independent} of the code distance $L$ when incorporating analog information into the decoder.}
    \label{fig:qss-3dtoric}
\end{figure*}

As a by-product, we obtain a threshold of the non-single shot side of the 3DSC under phenomenological noise of \mbox{$\approx 1.66\%$}.
In \appref{app:additional_results} we also present 3DSC threshold estimates for the case where only the hard syndrome information is available. 
We find that the threshold is $\approx 1.26\%$.
To the best of our knowledge, these are the first numeric threshold estimates for the non-single shot side of the 3DSC.

Let us make some more detailed remarks on the results. 
First, the results indicate that $w=2$ for the $L=11$ code leads to a worse threshold, indicating a limitation of the protocol in this aspect. 
However, the sub-threshold scaling (which is actually relevant in practice) still shows that the $2$-QSS protocol provides sufficient error suppression while lowering the number of syndrome measurement rounds from $11$ to $2$, inducing a fraction of the time overhead to implement the overall QEC protocol.
Secondly, increasing the QSS window size $w$ by $1$ already significantly improves the achieved logical error rate.
With $w = 5$ we do not find statistically significant differences from $w = L$ for the code sizes considered.
Lastly, we note that even for much smaller error rates, we did not find a threshold for $w = 1$.

\subsection{Discussion}\label{sec:qss_discussion}
For a QLDPC code that requires time-domain decoding, the effective distance is proportional to the number of repeated syndrome measurements.
Thus, taking only a small number $w \ll d$ of syndrome measurements is equivalent to lowering the effective code distance (along the time dimension) to $w$.
However, we show numerically that for reasonable code sizes and choices of the QSS window size $w$, the logical error rate of the overall protocol is equivalent to the standard approach of performing (at least) $d$ repeated measurements due to the additional information acquired from the analog syndrome.

Although our numerical results suggest that $w$ can be chosen as a constant independent of the code size for sufficiently small physical error rates, it is reasonable to assume that for larger code sizes (i.e., in the limit $n\to \infty$), the logical error rate for the QSS protocol diverges from the logical error rate obtainable from standard time-domain decoding (i.e., with a number of rounds proportionally to the code size) and will possibly result in an error floor set by time-like errors.
However, this gap in error suppression between the QSS protocol and the standard protocol quickly diminishes if the physical error rate is sufficiently below the threshold, as can be seen by inspecting~\figref{fig:qss-3dtoric}{b}.
For example, for $p_{\rm err}=0.009$, $w=2$ gives a significantly higher word error rate than the larger choices of $w$, but for smaller error rates, e.g., $p_{\rm err}=0.005$, the discrepancy with standard time domain decoding is reduced. 
Thus, overall the main learning is that we observe that for sufficiently small physical error rates, i.e., error rates sufficiently below threshold, the $w$-QSS protocol can give logical error rates that are equivalent to standard time-domain decoding.

It would be interesting to investigate this aspect analytically for an LDPC code family.
For instance, it is reasonable to assume that depending on the LLRs (weights) used for decoding one can argue that if the weights are $w = O(1)$ using hard information decoding and are increased to $\ell w$ by using analog information, one can reduce the number of repetitions by a factor of $\ell$ without affecting the logical error rate significantly, i.e., obtain an $L/\ell$-QSS protocol.
We leave further analytic investigation of this manner open for future work.
Note that the investigated codes are already well beyond the capabilities of near- to mid-term hardware, and thus we argue that our numerical results are valid for ``practical sizes''.

It is crucial to note that for the QSS protocol, the noise-biased error model is vital.
The main reason is that in this setting, the bulk of the decoding is offloaded onto the single-shot component of the 3DSC, while the QSS protocol carries a much lighter load.
This is also important due to the asymmetric thresholds of the 3STC for pure bit- and phase-flip noise, of approximately $1.5\%$ and $10\%$, respectively, which leads to the fact that the overall threshold of the code $p_{\rm err}^{\rm th}$ is limited by bit-flip errors.
However, according to Eq.~\eqref{eq:bias_eq}, already a small bias of $\eta_Z \approx 10$ will distribute the error correction load equally, and any bias $\eta_Z \gg 10$ will result in an effective bit-flip error rate $p_{\qgate{X}}$ that is well below the threshold of the QSS protocol if the overall error rate is $p_{\rm err} < 10\%$.

Although we are currently limited to simulations with codes of size $L \leq 11$ and physical error rates $p_{\rm err}$ around the threshold (due to the impracticality of conducting numerical experiments with larger codes), we argue that codes of such size are already reasonable for practical relevance due to good error suppression on the single-shot side of the code~\cite{quintavalle_single-shot_2021}.
However, for a more quantitative analysis, circuit-level noise model simulations are required. We expect circuit-level simulations will decrease observed thresholds by a factor of $4-8\times$. This estimate is consistent with the numerical results of Pattison et al., where two-dimensional surface codes were decoded under analog circuit-level noise~\cite{pattison_improved_2021}. 
Pattison \etal also proposed a generalization of the standard circuit-level noise model to include analog measurements to simulate surface code decoding under realistic noise assumptions.
Since the generalization to circuit-level noise encompasses several non-trivial questions such as the design of syndrome extraction circuits and the derivation of exact cat qubit noise models, we leave this task open for future work. 
Nonetheless, we discuss several challenges in more detail in \secref{sec:architectures}.

To verify that the QSS protocol does not lead to a decrease of the logical error rate (and the threshold) for an increasing number of decoding rounds, we conducted sustained threshold simulations, i.e., threshold simulations for an increasing number of decoding rounds.
However, due to the saturation of the logical error rates in the numerical results, we are only able to obtain a lower bound on the threshold, which decreases with the number of decoding rounds. 
Therefore, we provide additional results, shown in~\figref{fig:qss_logical_decay}, which demonstrate that the decay of the decoding success (logical success rate) for a number of decoding rounds scales equivalently for the QSS protocol compared to the standard time-domain decoding, where the number of syndrome extraction rounds is $w=L$ (i.e., proportional to the distance) for physical error rates below the obtained threshold lower bounds. 
Despite the fact that this result constitutes a weaker statement than a sustained threshold estimate, it indicates that the performance of the QSS protocol is equivalent to standard time-domain decoding even for an increasing number of decoding rounds.

\begin{figure}[t]
    \centering    \includegraphics{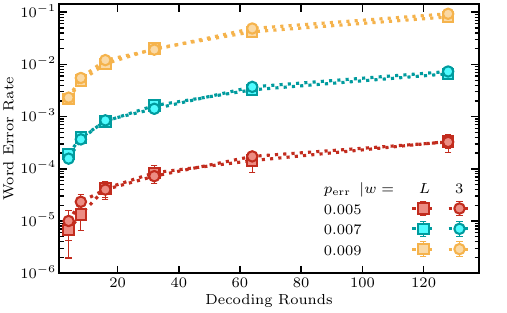}
    \caption{Comparison of the decrease of decoding success of the $3$-QSS decoding ($w=3$) versus the standard time-domain decoding ($w=L$) for an increasing number of decoding rounds and sub-threshold physical error rates using the $L=13$ three-dimensional surface codes under phenomenological noise.}
    \label{fig:qss_logical_decay}
\end{figure}

Moreover, we emphasize that while our simulations indicate that the QSS protocol works in a memory experiment, it is an open question whether this is also the case for a setting in which fault-tolerant logical gates are implemented. Another feature that could potentially affect decoding performance is the presence of so-called fragile boundaries, as discussed in Ref.~\cite{higgott_improved_2023-1}.

\section{Towards three-dimensional concatenated cat codes}\label{sec:architectures}
In this section, we discuss roads towards building a fault-tolerant quantum computer based upon stabilized cat qubits concatenated with the three-dimensional surface code.
To this end, we will elaborate on several open questions and challenges along that road.
We begin by recalling some properties of stabilized cat qubits.

\subsection{Stabilized cat codes}\label{sec:cats}

The cat code encodes logical (qubit) information within a two-dimensional subspace of the infinite-dimensional Hilbert space of a harmonic oscillator with Hilbert space ${\cal L}(\R^2)$.
This qubit subspace is represented by the span $\lbrace \ket{\alpha}, \ket{-\alpha} \rbrace$ of two quasi-orthogonal coherent state vectors $\ket{\pm \alpha}, \alpha \in \mathbb{C}$~\cite{walls_quantum_2008} (in the sense that for large values of $|\alpha|$, they approximate an orthogonal pair of state vectors arbitrarily well).
The non-orthogonality of this basis does not pose a problem for the definition of the orthogonal qubit space, and 
one defines the Hadamard-dual basis codewords $\ket{\pm}_{\qcode{cat}}$ as two-component Schr{\"o}dinger cat state vectors, i.e., 
\begin{equation}
    \label{eq:cat_dual_basis_states}
    \ket{\pm}_{\qcode{cat}} = \mathcal{N}_{\pm} (\ket{\alpha} \pm \ket{-\alpha}),
\end{equation}
which are orthogonal state vectors, and 
\begin{equation}
\mathcal{N}_{\pm}^2 := 1 / {(2 (1 \pm e^{-2\vert \alpha \vert^2}))}
\end{equation}
is the normalization factor\footnote{
Intuitively, one can recognize that these states are orthogonal by noting that $\ket{+}_{\qcode{cat}}$ is invariant under the exchange $\alpha \leftrightarrow - \alpha$ while $\ket{-}_{\qcode{cat}}$ obtains a global phase. This makes them $\pm 1$ eigenstates of the parity operator $\hat{\Pi} = \exp(i \pi \hat{a}^{\dagger} \hat{a})$, respectively, and thus orthogonal.}.
Then, the logical, computational state vectors are obtained as
\begin{align}
    \label{eq:cat_comp_states}
    \ket{0}_{\qcode{cat}} = \frac{1}{\sqrt{2}} (\ket{+}_{\qcode{cat}} + \ket{-}_{\qcode{cat}}) =\ket{+ \alpha} + O(e^{-2\vert \alpha \vert^2}), \\
    \ket{1}_{\qcode{cat}} = \frac{1}{\sqrt{2}} (\ket{+}_{\qcode{cat}} + \ket{-}_{\qcode{cat}}) = \ket{- \alpha} + O(e^{-2\vert \alpha \vert^2}),
\end{align}
and the approximations $\ket{0}_{\qcode{cat}} \approx \ket{\alpha}$ and $\ket{1}_{\qcode{cat}} \approx \ket{-\alpha}$ become arbitrarily 
accurate for $\vert \alpha \vert^2 \rightarrow \infty$.

The cat code space is not stable under noise channels that typically affect the physical realizations of harmonic oscillators---dominated by energy relaxation reflected by losses and dephasing. Thus, any logical information will eventually leak outside of the code space and will be unrecoverable.
However, through engineered interactions, it is possible to stabilize the code space through appropriate confinement schemes. 
While various different confinement schemes, such as 
Kerr stabilization~\cite{puri_engineering_2017}, dissipative stabilization~\cite{mirrahimi_dynamically_2014}, and combined methods~\cite{gautier_combined_2022}, exist, they share similar principles. 
First, to overcome energy relaxation, one actively pumps energy into the 
system through engineered (two-photon) drives.
Then, an actual ``confinement'' term is added that separates the cat qubit manifold from the rest of the energy spectrum.
To ensure a two-fold degenerate ground state of the system, these engineered interactions must be symmetric with respect to the substitution $\hat{a} \mapsto - \hat{a}$, where $\hat{a}$ is the bosonic annihilation operator satisfying the canonical commutation relation $\comm*{\hat{a}}{\hat{a}^{\dagger}} = \Id $~\cite{albert_lindbladians_2018}.

Stabilization through some confinement interaction ensures that if leakage occurs, the state will relax back to the code space.
As a result,  stabilized cat codes allow for arbitrary suppression of bit-flip noise under realistic oscillator noise models\footnote{
It should be emphasized that state-of-the-art experiments with dissipatively stabilized cat qubits cannot arbitrarily suppress bit-flip errors, however. The current understanding is that this is due to additional noise channels caused by auxiliary (few level) qubits, e.g., transmon qubits, used for control and readout~\cite{berdou_one_2023}.}~\cite{puri_bias-preserving_2020, guillaud_repetition_2019}, as we will illustrate in the case of single-photon losses below.
However, the confinement does not protect against logical cat qubit Z errors on the code space, which may occur directly through oscillator decoherence, such as phase-flips caused by energy relaxation, or indirectly if a noise channel leads to temporary leakage out of the code space, e.g., caused by thermal noise.

The time evolution of a single-mode quantum system undergoing single-photon loss is well described by the 
Lindblad master equation~\cite{lindblad_generators_1976, gorini_completely_1976},
\begin{equation}
    \label{eq:single_photon_dissipator}
    \frac{\partial}{\partial t} \hat{\rho} =  \kappa \mathcal{D}[\hat{a}] \hat{\rho} = \frac{\kappa}{2} \left(2 \hat{a} \hat{\rho} \hat{a}^{\dagger} - \hat{a}^{\dagger} \hat{a} \hat{\rho} - \hat{\rho} \hat{a}^{\dagger} \hat{a}  \right),
\end{equation}
where $\kappa>0$ is the single-photon loss rate and $\hat{\rho}$ is the density operator describing the state of the system. 
Here, the first term leads to quantum jumps, whereas the latter two terms generate a non-Hermitian evolution that leads to energy relaxation.
One can calculate the leading-order estimates for the cat qubit phase- and bit-flip error rates through the Knill-Laflamme conditions~\cite{nielsen_quantum_2010, knill_theory_1997} for the oscillator error $\hat{E}_1 \propto \sqrt{\kappa} \hat{a}$.
The task reduces to the following transition matrix elements
\begin{align}
    \kappa \vert \mel{+\alpha}{\hat{a}}{-\alpha} \vert^2 &= \kappa \vert \alpha \vert^2 e^{- 4 \vert \alpha \vert^2}, \\
    \kappa \vert \mel{+}{\hat{a}}{-}_{\qcode{cat}} \vert^2 &= \kappa \vert\alpha\vert^2 \tanh(\vert\alpha\vert^2) \approx \kappa \vert \alpha \vert^2,
\end{align}
which shows that the ratio of bit- to phase-flip errors is exponentially suppressed, i.e.,  $p_x / p_z \sim e^{-4 \vert \alpha \vert^2}$, yielding an effective biased-noise error channel for the outer code.
We note that exponential suppression (in $\vert \alpha\vert^2$) of bit-flip errors comes at the cost of linearly increasing the phase-flip rate.
Although this will limit the extent to which one can increase $\alpha$ before $p_z$ exceeds the threshold of the concatenated code, we emphasize that the recently introduced \emph{stabilized squeezed-cat qubit} allows one to suppress bit-flip errors without increasing the phase-flip errors~\cite{hillmann_quantum_2023}.

To benefit from a biased-noise error channel, it is important that the noise bias can be sustained even during gate operations. 
It has been shown that this is possible for stabilized cat qubits due to the complex-valued displacement amplitude $\alpha$, which contributes additional degrees of freedom and in this way allows the realization of the two-qubit \qgate{CNOT} gate in a bias-preserving way by performing (conditioned) rotations that exchange $\ket{\alpha}$ and $\ket{-\alpha}$. 
During this rotation, the bias is preserved, and a topological phase is added to the dual-basis codewords, i.e.,  $\ket{+}_{\qcode{cat}} \mapsto \ket{+}_{\qcode{cat}}$ and $\ket{-}_{\qcode{cat}} \mapsto \ket{-}_{\qcode{cat}}$~\cite{puri_bias-preserving_2020, guillaud_repetition_2019}.
In a circuit-level noise model the ratio of single photon losses and confinement rate is an important parameter that relates effective encoded cat qubit Pauli error rates to physical noise parameters, see, e.g., Refs.~\cite{guillaud_repetition_2019, chamberland_building_2022, darmawan_practical_2021} for more details.

Finally, performing a logical \qgate{Z} measurement can be done, for example, by performing a non-demolition cat quadrature readout~\cite{grimm_stabilization_2020}, which distinguishes the two coherent state vectors $\ket{+ \alpha}$ and $\ket{- \alpha}$, see also~\sfigref{fig:fig1}{b}.
Due to the finite variance of coherent states, such a measurement will be inherently imprecise because of their continuous 
distribution in quantum phase space.
However, one can incorporate this analog information into the decoding stages of the outer code, for 
example, by assigning higher error likelihoods to states that have measurement outcome $x_m \approx 0$.
Importantly, the \emph{resolvability} of the cat qubit computational state vectors $\ket{0 / 1}_{\qcode{cat}}$ is given by the overlap of the two states that scales as $\sim e^{-2 \vert \alpha \vert^2}$ and thus assignment errors become exponentially suppressed with the size of the stabilized cat qubit.

\subsection{Open questions}

We highlight that an immediate open question, independent from any experimental realization, is the verification of our decoding protocols in a more realistic noise model, i.e., in the presence of circuit-level noise and the determination of thresholds in these cases.
One might expect a reduction in threshold (roughly) proportional to the stabilizer weight, due to additional fault locations that occur in the syndrome extraction circuit, impacting the non-single shot side of the code more strongly than the single shot side, which have stabilizers of weights 6 and 4, respectively. 
However, this will not cause a fundamental issue, as the bias of the stabilized cat qubits can be tuned such that the code performance is effectively limited by the threshold of the single-shot code.
Regarding syndrome extraction, very recent work suggests that the ordering of operations in the syndrome extraction circuit does not affect the effective distance of the code, see Ref.~\cite{manes_distance-preserving_2023}.

Syndrome extraction based on cat qubits requires bias-preserving \qgate{CNOT} gates, which have not been demonstrated in experiments for stabilized cat qubits so far.
Therefore, currently, our estimates for achievable error rates with such gates rely upon theoretical models as proposed, for instance, in Ref.~\cite{darmawan_practical_2021} for Kerr-cats and Ref.~\cite{chamberland_building_2022} for dissipative cats.
These references also detail the implementation of all other required Clifford operations and Pauli measurements required for the two-dimensional surface code.
As there is no fundamental difference in the type of gates required for syndrome extraction in the three-dimensional case, we refer the interested reader to the aforementioned articles.

\subsection{Conceptional architecture}

One could imagine a possible hardware implementation in superconducting circuits as an extension of the proposals in Refs.~\cite{darmawan_practical_2021, chamberland_building_2022}, stacking the proposed two-dimensional layouts in a vertical direction in 
an alternating ABAB pattern as illustrated in~ \figref{fig:architecture_sketch}.
Vertical coupling between chips can be achieved through small form factor superconducting \emph{through-silicon-vias} (TSVs)~\cite{yost_solid-state_2020, mallek_fabrication_2021, grigoras_qubit-compatible_2022}.
Although state-of-the-art fabrication techniques currently do not achieve stacking of more than a few layers, the use of TSVs in superconducting circuits is a recent development that will likely mature rapidly in the future~\cite{hazard_characterization_2023}.
Conceptually, even only a few layers can yield a useful three-dimensional surface code when the noise bias is large enough. 
The reason is that for the rectangular cubic lattice of spatial extend $L_x, L_y,$ and $L_z$, the effective code distances $d_X$ and $d_Z$ are given by~\cite{quintavalle_single-shot_2021}
\begin{align}
    d_X &= \min \{L_x, L_y, L_z\}, \\
    d_Z &= \min \{L_x L_y, L_y L_z, L_z L_x\}.
\end{align}

\begin{figure}[t]
    \centering
    \includegraphics{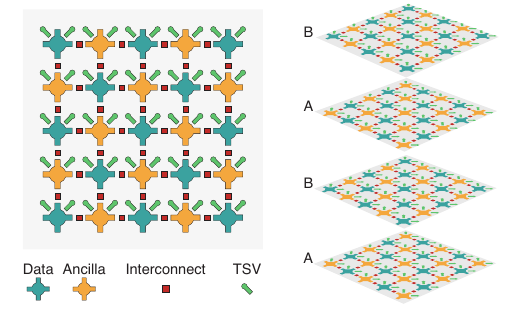}
    \caption{Sketch of the possible three-dimensional surface codes architecture. The grid on the left shows a single two-dimensional layer of data and auxiliary cat qubits arranged with nearest-neighbor connectivity achieved through in-plane interconnects that activate the interaction between the data and auxiliary qubits.
    \emph{Through-silicon-vias} (TSVs) connect multiple such layers together, one connecting to the layer above, the other to the layer below. The layers are stacked in an ABAB pattern (right), where the difference between A and B is that the placement of data and auxiliary qubits is interchanged.}
    \label{fig:architecture_sketch}
\end{figure}

\pagebreak

\section{Conclusion}\label{sec:conclusion}
Recently, there has been progress in the realization of bosonic codes that increase the lifetime of encoded quantum information. 
Additionally, the discovery of \emph{good} QLDPC codes~\cite{hastings_fiber_2021, breuckmann_balanced_2021, panteleev_asymptotically_2022, dinur_good_2022, leverrier_quantum_2022} motivates the development of decoding protocols for concatenated bosonic-LDPC codes.
These protocols should consider the analog information inherent to the measurement of continuous-variable quantum states.

In this article, we contribute to this task by presenting methods that feed the analog information obtained during bosonic syndrome measurements into belief propagation and matching decoders. 
In particular, we show how to decode analog syndromes for single-shot codes that are obtained from higher-dimensional hypergraph product constructions. We also consider codes that are not single-shot and thus require repeated stabilizer measurements over time in general.
We introduce \emph{analog Tanner graph decoding} as a way of naturally incorporating analog syndrome information directly into the decoding graph.

To support and numerically assess our decoding methods, we consider the three-dimensional surface code as a test case. Our simulations are performed using a phenomenological noise model inspired by bosonic cat code qubits.
We find that our analog Tanner graph decoding methods lead to a significantly enhanced sustainable single-shot threshold for the three-dimensional surface code.
Furthermore, we show that accounting for analog information from bosonic syndrome measurements can reduce the number of repetitions required for time-domain decoding. 
We demonstrate this explicitly by incorporating analog Tanner graph methods into an overlapping window decoder for the non-single-shot component of the three-dimensional surface code. 
For the case of the $L=13$ three-dimensional surface code, we show that it suffices to decode with a window size of $w=3$. 
This is a considerable reduction in the time overhead compared to the case of discrete syndrome decoding where the window size must be equal to the code distance, i.e., $w=13$. 
We argue that this renders the three-dimensional surface code \emph{w-quasi-single-shot}.

To further boost the development of concatenated bosonic-LDPC codes, we provide open-source software tools for all proposed techniques.
With these tools, we hope to emphasize the importance of open-source software and to inspire further research interest into concatenated bosonic codes.

We note that our numerical experiments are performed using a phenomenological noise model.
A natural follow-up to this work will be to further verify the potential of analog Tanner graph decoding under a more realistic circuit-level noise model and to investigate the QSS protocol and the use of analog information in general analytically using appropriate metrics~\cite{danjou_generalized_2021}.
Moreover, it is an interesting question as to whether other decoders for (3D) QLDPC codes, such as the recently introduced ``p-flip decoder''~\cite{scruby_local_2023} or the three-dimensional tensor network decoder~\cite{piveteau_tensor_2023}, can be modified to incorporate analog information into the decoding process. 
Finally, it would also be interesting to investigate the performance of analog Tanner graph decoding for other codes, such as three-dimensional subsystem codes~\cite{kubica_single-shot_2022, bridgeman_lifting_2023}.

This paper has focused on quantum memories. A remaining open problem concerns the questions as to whether analog information can be used to improve decoding performance during the implementation of fault tolerance logical gates, e.g., during lattice surgery. 
Such investigations will require a detailed analysis of the physical architecture used to realize the bosonic qubits. For instance, a necessary requirement will be that the qubits support bias-preserving two-qubit gates~\cite{puri_bias-preserving_2020}.

\begin{acknowledgments}
The authors thank Oscar Higgott, Nithin Raveendran, and Fernando Quijandría for valuable comments on the first draft of this manuscript and Armanda O. Quintavalle for fruitful discussions and comments.
We thank anonymous referees for numerous valuable comments on the first version of this manuscript. 
L.B. and R.W. acknowledge funding from the European Research Council (ERC) under the European Union’s Horizon 2020 research and innovation program (grant agreement No. 101001318) and MILLENION, grant agreement No. 101114305). This work was part of the Munich Quantum Valley, which is supported by the Bavarian state government with funds from the Hightech Agenda Bayern Plus, and has been supported by the BMWK on the basis of a decision by the German Bundestag through project QuaST, as well as by the BMK, BMDW, and the State of Upper Austria in the frame of the COMET program (managed by the FFG).
T.H. acknowledges the financial support from the Chalmers Excellence Initiative Nano and the Knut and Alice Wallenberg Foundation through the Wallenberg Centre for Quantum Technology (WACQT). 
J.R. and J.E.~are funded by BMBF (RealistiQ, QSolid), 
the DFG (CRC 183), 
the Quantum Flagship (Millenion, PasQuans2),
the Einstein Foundation (Einstein Research Unit on Quantum Devices), and
the Munich Quantum Valley (K-8). 
J.E.~is also funded by the European Research Council (ERC) within the project DebuQC.
L.B.~and T.H.~thank IBM for hosting the 2022 QEC summer school, where initial ideas for this project have been developed.

\end{acknowledgments}

\appendix

\section{$\Ft$-homology}\label{sec:app-homology}
CSS codes are equivalent to 3-term chain complexes of binary vector spaces. 
A chain complex of vector spaces $(C_{\sbullet}, \bdry_{\sbullet})$ is a sequence of vector spaces and linear maps
\begin{equation}\label{eq:chain-complex}
	(C_{\sbullet}, \bdry_{\sbullet}) = \dots C_{i+1} \xrightarrow{\bdry_{i+1}} C_i \xrightarrow{\bdry_{i}} C_{i-1} \xrightarrow{\bdry_{i-1}} \dots,
\end{equation}
with the property
\begin{equation}\label{eq:cc-condition}
	\bdry_{i}\bdry_{i+1} = 0, \forall i.
\end{equation}
The linear maps $\bdry_i$ are called \emph{boundary maps} or \emph{boundary operators}.
It is standard to define the spaces of \emph{cycles} $Z_i$ and \emph{boundaries} $B_i$ as
\begin{align}
	Z_i &:= \text{ker}\ \bdry_{i} \subseteq C_i ,\\
	B_i &:= \text{im}\ \bdry_{i+1} \subseteq C_{i}.	
\end{align}
Since Eq.~\eqref{eq:cc-condition} implies that $Z_i \subseteq B_i$, we can define the quotient
\begin{equation}
	\calH_i(C_{\sbullet}) := Z_i/B_i,
\end{equation}
which is called the $i$-th \emph{homology group} of the chain complex. 

By inverting the arrows in Eq.~\eqref{eq:chain-complex}, i.e., transposing the corresponding linear maps, we obtain the dual notion called \emph{co-chain complex}
\begin{equation}\label{eq:co-chain-complex}
C^{\sbullet} := \dots C_{i+1} \xleftarrow{\bdry^{\top}_{i}} C_i \xleftarrow{\bdry^\top_{i-1}} C_{i-1} \dots,
\end{equation}
and completely analogous definitions for \emph{co-boundaries} $B^i := \text{im}\ \bdry_{i-1}^{\top}$, \emph{co-cycles} $Z^i := \text{ker}\ \bdry^{\top}_i$, and the \emph{co-homology group} $\calH^i(C^{\sbullet}) = Z^i/B^i$.

By using a three-term subcomplex of a chain complex
\begin{equation}
	C_{i+1} \xrightarrow{\bdry_{i+1}} C_i \xrightarrow{\bdry_i} C_{i-1}, 
\end{equation}
a CSS code can be obtained by setting
\begin{align}
	H_Z^T &= \bdry_{i+1}, \\
	H_X &= \bdry_{i},
\end{align}
whereby the CSS condition from Eq.~\eqref{eq:css-orthogonality} is fulfilled by definition.
We can now reason about a code in the language of chain complexes and their homology.

The group generated by the \qgate{Z}-type stabilizers $S_Z$ correspond to the boundaries $B_i$ and the \qgate{Z}-type Pauli operators that commute with all \qgate{X}-type stabilizers correspond to the cycles $Z_i$. 
Analogously, $S_X = B^i$ and the \qgate{X}-type Paulis commuting with the \qgate{Z}-type stabilizers correspond to the co-cycles $Z^i$.
The \qgate{Z}-type logical operators correspond to elements of the homology group $\calH_i$ and the \qgate{X}-type logical operators to the cohomology group $\calH^i$.

Note that a linear classical code is a two-term chain complex where the boundary operators map between the space of checks and the code space.

Using the language of homology, codes can be constructed by taking the product of two chain complexes~\cite{tillich_quantum_2014, hastings_fiber_2021, breuckmann_balanced_2021, panteleev_asymptotically_2022}. 
The tensor product\footnote{Formally this is denoted as the total complex of the tensor product double complex~\cite{breuckmann_balanced_2021}.} of two $2$-term chain complexes $C_1 \xrightarrow{\bdry^C_1} C_0$ and $D_1 \xrightarrow{\bdry^D_1} C_0$, each corresponding to a classical code, gives a three-term chain complex $C\otimes D$ defined as
\begin{align}\label{eq:cc-tensor-prod}
	C_1 \oplus D_1 \xrightarrow{\bdry_2} 
	C_1 \oplus D_0 \otimes C_0 \oplus D_1 \xrightarrow{\bdry_{1}}
	C_0 \otimes D_0,
\end{align}
where the boundary maps are defined as
\begin{align}
	\bdry_{2} &= \left(
	\begin{matrix}
		\bdry_{1}^C \otimes \Id \\
		  \Id \otimes \bdry_{1}^D
	\end{matrix}
	\right),\\
	\bdry_{1} &= \left(
	\bdry_{1}^C \otimes \Id \mid \Id \otimes \bdry_{1}^D
	\right).
\end{align}
Applying the tensor product to higher-dimensional chain complexes gives a quantum code with higher-dimensional elements. 
 
For example, the repetition (ring) code can be seen as a collection of vertices connected by edges in pairs.
\begin{figure}[!b]
    \centering
    \includegraphics{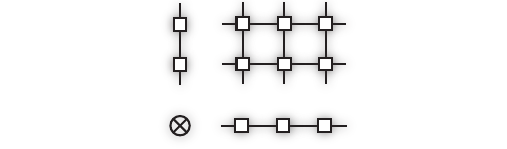}
    \caption{The two-dimensional surface code obtained from the tensor product complex of two repetition codes.}
    \label{fig:2dsc-cc-tensor-product}
\end{figure}
The tensor product of two repetition codes then describes a two-dimensional object with faces, edges, and vertices that correspond to the two-dimensional surface code, as illustrated in~\figref{fig:2dsc-cc-tensor-product}.
Analogously, the three-dimensional surface codes~\cite{vasmer_three-dimensional_2019} can be obtained as a tensor product of a two-dimensional surface code with a repetition code corresponding to a 4-term chain complex (cf.~Ref.~\cite{strikis_quantum_2023}) as
\begin{equation}\label{eq:3d-chain-complex}
	C_3 \xrightarrow{\bdry_{3}} C_2 \xrightarrow{\bdry_{2}} C_1 \xrightarrow{\bdry_{1}} C_0.
\end{equation}
\begin{example}[Three dimensional surface code from repetition codes]
    Consider the $3$-repetition code $R: C_1 \to C_0$ with check matrix 
        \begin{equation}\label{eq:rep-code}
            H_{\rm rep} = \left(
                    \begin{matrix}
                        1 & 1 & 0 \\
                        0 & 1 & 1 \\
                        1 & 0 & 1
                    \end{matrix}
                \right).
        \end{equation}
        A two-dimensional surface code can be obtained by taking $S = R \otimes R$. 
        By taking the tensor product with $R$ again, we obtain a three-dimensional surface code, i.e., a three-dimensional lattice $S_{3D} = S \otimes R$, as sketched in~\figref{fig:3dsc}. 
\end{example}

Depending on whether we choose period boundary conditions---i.e., a ring code or a repetition code as ``seed code''---or not, we obtain the following code parameters of the three-dimensional surface codes (3DSC). Note that the three-dimensional surface codes with periodic boundaries is also called \emph{three-dimensional toric code} (3DTC),
\begin{itemize}
	\item 3DSC: $[[2L(L-1)^2+L^3, 1, d_X=L^2, d_Z=L]]$,
	\item 3DTC: $[[3L^3, 3, d_X=L^2,d_Z=L]]$.
\end{itemize}
Note that, instead of placing \qgate{X}-checks on faces and \qgate{Z}-checks on vertices, some works consider an assignment with swapped checks. 

\section{Proof of~\Cref{prop:informal}}\label{ssec:proofs}
Here, we argue that the construction of the multi-round parity-check matrix $\tilde{H}$ (cf.~\Cref{def:mr-pcm}) for $r$ rounds of syndrome measurement can be described as the tensor product of the chain complex of the code $\calC$ and the chain complex of a (slight variant of the) $r$-repetition code $\calR$.

The statement is quite straightforward given existing results on product code constructions and hence the result follows from basic notions from graph theory and homological algebra.  
However, technically, it is a priori not clear that this matches our multi-round Tanner graph construction, thus to make these correspondences concrete, we present the result formally in the following. 

Let us introduce some additional notation.
We consider a two-dimensional \emph{space} $X$ as the generalization of a graph (a one-dimensional space) with the $i$-cells, $X_i$, denoting sets of the $i$-dimensional elements, i.e., $0$-cells are vertices, $1$-cells are edges, $2$-cells are faces, and $3$-cells are volumes. 
Analogously to graphs, the \emph{incidence matrices} $\bdry_i$ of $X$ are defined as
\begin{equation}
	\bdry^X_{i} \in \Ft^{X_i}, (\bdry^X_i)_{v,w}=1 \iff v \sim w,
\end{equation}
where $v\sim w$ denotes that $v$ is incident to $w$.

Given a two-dimensional space $X$, the \emph{cellular chain complex} $C_{\sbullet}(X)$ is defined as the chain complex whose vector spaces $C_i$ have the $i$-cells, $X_i$, as basis and boundary maps $\bdry_i$ that map an $i$-cell to the formal sum of $(i-1)$-cells at its boundary, e.g., an (edge) to the vertices at its boundary
\begin{equation}
	C_{\sbullet} = C_{2}(X) \to C_{1}(X) \to C_{0}(X)
\end{equation}
where we may identify $C_{i}(X) = \Ft^{X_{i}}$, i.e., $C_{i}(X)$ is the vector space spanned by $i$-cells in $X$, and the boundary operators correspond exactly to the incidence matrices of the $i$-space.
Given two graphs (one-dimensional spaces) $X,Y$, their Cartesian product $X\times Y$ is a $2$-dimensional space $Z$ whose elements are
\begin{align}
	Z_0 &= X_0 \times Y_0, \\
	Z_1 &= X_0 \times Y_1 \coprod X_1 \times Y_0,  \\
	Z_2 &= X_1 \times Z_1,
\end{align}
where the coproduct $\coprod$ is the disjoint sum of sets.
The incidence matrices of $Z$ are then given as
\begin{align}
	\bdry_1^Z &= (\Id_{X_0} \otimes \bdry^Y_1 \mid \bdry_1^X \otimes \Id_{Y_{0}}) \in \Ft^{Z_1}, \\
	\bdry_2^Z &= \left(
			\begin{matrix}
				\bdry^X_1 \otimes \Id_{Y_{1}} \\
				\Id_{X_{1}} \otimes \bdry^Y_{1}
			\end{matrix}
		\right) \in \Ft^{Z_2}.
\end{align}
Note that this is equivalent to considering the cellular chain complexes $C_{\sbullet}(X)$ and $D_{\sbullet}(Y)$ and constructing the tensor product complex 
\begin{equation}
	C_{\sbullet}(X)\otimes D_{\sbullet}(Y).
\end{equation}
Since $C_{\sbullet}(X), D_{\sbullet}(Y)$ come with bases $X_0, X_1$ and $Y_0, Y_1$, the bases of $C_{\sbullet}(X)\otimes D_{\sbullet}(Y)$ correspond exactly to the spaces obtained by the Cartesian product $X\times Y$ and by ordering the Cartesian products, the matrices of the boundary operators $\bdry_i^Z$ are exactly the Kronecker products of the corresponding matrices $\bdry_i^X, \bdry_i^Y$, hence the boundary map of the tensor product complex are exactly the incidence matrices of the Cartesian product and we can identify 
$C_{\sbullet}(X)\otimes D_{\sbullet}(Y) = C_{\sbullet}(X\times Y)$.
Having the notation in place we can formulate the statement:

\begin{proposition}[Formal version of Proposition~\ref{prop:informal}]\label{prop:mr-tg-is-cc-prod}
    Let $C_{\sbullet}$ be a three-term chain complex corresponding to a CSS QLDPC code $\calC$ and let $R_{\sbullet} = R_1 \xrightarrow[]{R} R_0$, be a chain complex whose boundary map $R$ corresponds to the $r\times r$ matrix
    \[R = \left(\begin{matrix}
        1 & 0 & \dots &0 \\
        &  & R'  &
    \end{matrix}\right),\]
    where $R'$ is the $(r-1)\times r$ check matrix of the $r$-repetition code.
    Then, the $r$-multi-round parity check matrix $\tilde{H}$ of $\calC$ is equivalent to the boundary map of the tensor product complex $R_{\sbullet} \otimes C_{\sbullet}$.
\end{proposition}
\begin{proof}
Since the code is CSS we focus on a single check side (i.e., the underlying graph of the space) in the following. 
Let $H \in \Ft^{m \times n}$ be the parity check matrix of one side of the code $\calC$, i.e., 
\[C_1 \xrightarrow[]{H} C_0.\]
The tensor product chain complex is
\begin{equation}
    R_1 \otimes C_1 \xrightarrow[]{\bdry_2} R_0 \otimes C_1 \oplus R_1 \otimes C_0 \xrightarrow[]{\bdry_1} R_0 \otimes C_0,
\end{equation}
where the boundary maps are given by 
\begin{align}
    \bdry_2 &= \left( 
        \begin{matrix}
            \Id_R \otimes H \\
            R \otimes\Id_H
        \end{matrix}
    \right), \\
    \bdry_1 &= \left( 
         \Id_R\otimes H \mid R \otimes \Id_H
    \right).
\end{align}

Since qubits are placed on $1$-cells, the check matrix given by $\bdry_1$ is the one that is relevant. 

Viewing $C_{\sbullet}, R_{\sbullet}$ as cellular chain complexes, it is clear that the basis elements of the $1$-cells correspond exactly to $Y_0 \times X_1 \coprod Y_0 \times X_1$, where $Y_i, X_i$ are the bases of the $i$-cells of $R$ and $C$, respectively, and $\bdry_1$ is exactly the incidence matrix of the underlying $1$-complex. 

Hence, by the definition of the Kronecker product, the resulting check matrix has the form
\begin{equation}
    \bdry_1 = \left(
        \begin{matrix}
            H &         &    & & \Id_H &  &  \\
                   &  H &  &  &  \Id_H & \Id_H &  & \\
                     &     &\ddots &            &  & \ddots  & \\
                     &     &  &  H & & &\Id_H & \Id_H
        \end{matrix}
    \right).
\end{equation}
Thus, $\bdry_1 = \tilde{H}(\calC)$. 
Since the edge-vertex incidences are given by $\bdry_1$ in the corresponding graph whose edges can be identified with the bases $R_0\times C_1 \coprod R_1\times C_0$, and whose vertices can be identified with $R_0\times C_0$, the product graph obtained is equivalent to the multi-round Tanner graph $\tilde{\calT}$.
\qedhere

Note that to match~\Cref{def:mr-pcm} exactly, we consider a slightly altered check matrix $R$ compared to the standard repetition code. 
For example, the check matrix of the $4$-repetition code is 
\begin{equation}\label{eq:4rep-code}
    R_{4-{\rm rep}} = \left( 
        \begin{matrix}
            1 & 1 & 0 & 0\\
            0 & 1 & 1 & 0\\
            0 & 0 & 1 & 1\\
        \end{matrix}
    \right),
\end{equation}
the version we consider is 
\begin{equation}\label{eq:4-rep-code}
    R = \left(
        \begin{matrix}
            1 & 0 & 0 & 0 \\
            1 & 1 & 0 & 0 \\
            0 & 1 & 1 & 0 \\
            0 & 0 & 1 & 1
        \end{matrix}
    \right),
\end{equation}
as this accounts for the fact that the first layer of checks is only connected to a single layer of time-like bit nodes.
Note that the code spaces of the matrices defined above are not equivalent, since for the repetition code from Eq.~\eqref{eq:4rep-code} the all-ones vector is the only non-trivial codeword $(1,1,1,1) \in ker(R_{4-{\rm rep}})$, but for the considered variant defined in Eq.~\eqref{eq:4-rep-code} we have $(0,1,1,1) \in \mathit{ker}(R)$. 
One could equivalently consider the standard repetition code matrix and then project the final boundary map s.t.\ the respective identity block entry is mapped to 0.
\end{proof}

\section{Implementation details}\label{sec:app_impl-details}
In this section, we present details concerning the code used to conduct the numerical experiments presented in this manuscript. In ~\appref{sec:app_code_constructions} we review the QLDPC code family used in~\secref{sec:atd}. \appref{sec:s-noise-model-conversion} reviews the conversion between analog syndrome noise and bit-wise syndrome noise channels used in~\secref{sec:atd}. In~\appref{sec:app_bp} to \appref{sec:app_atd-details} we review details on belief-propagation decoding and the proposed implementations of ATD and SSMSA.

\subsection{Non-topological code constructions}\label{sec:app_code_constructions}
In this section, we give details on the construction of the codes used for numerical evaluations.

\subsubsection{Lifted product codes}\label{sec:app-lp-codes}
For the simulations presented in~\secref{sec:decoding-main}, we use a family of \emph{lifted product} (LP) codes~\cite{panteleev_degenerate_2021, panteleev_asymptotically_2022}. The construction of lifted product codes is described below.
Algebraically, an $[[n,k,d]]$ LP code can be obtained from the tensor product of a \emph{base} matrix $B$ that corresponds to a classical quasi-cyclic LDPC code~\cite{fossorier_quasicyclic_2004} with its conjugate transpose $B^*$. 
The concrete instances of the family used are constructed from the base matrices $B_d$ for distance $d$ from Appendix A in Ref.~\cite{raveendran_finite_2022} and are given as 
\begin{align}\label{eq:lp-base-mat}
	B_{12} &= \left(
	\begin{matrix}
		0 &0 &0 &0& 0 \\
		0 &2 &4 &7& 11\\
		0& 3 &10 &14& 15 \\
	\end{matrix}
	\right), \\
	B_{16} &= \left(
	\begin{matrix}
		0 &0 &0 &0& 0 \\
		0 &4& 5 &7 &17 \\
		0 &14 &18 &12& 11
	\end{matrix}
	\right), \\
	B_{20} &= \left(
	\begin{matrix}
		0 &0 &0 &0& 0 \\
		0 &2 &14 &24& 25 \\
		0 &16 &11 &14& 13
	\end{matrix}
	\right),
\end{align}
to obtain the code instances with parameters
\begin{itemize}
	\item $[[544,80, d\leq 12]]$,
	\item  $[[714,100, d\leq 16]]$,
	\item $[[1020, 136, d\leq 20]]$.
\end{itemize}
To construct the code instances in software, we use the LDPC library by Roffe~\cite{roffe_ldpc_2022}. 
The parity-check matrices are provided in the GitHub repository~\url{github.com/cda-tum/mqt-qecc}.

\subsection{Syndrome noise model conversion}\label{sec:s-noise-model-conversion}
To compare decoding approaches that consider analog or hard syndrome errors, the syndrome noise model under consideration needs to be compatible. This means we need to be able to compare (and convert) the strength of the analog syndrome noise to the hard syndrome noise and vice versa. 
The analog information decoder considers Gaussian syndrome noise $e_i \sim \calN(0, \sigma^2)$. 
When dealing with syndrome bits $s_i \in \set{-1,+1}$, we want to convert this into an error channel for hard syndrome noise that is equivalent, i.e., 
\begin{equation}\label{eq:gauss-noise-to-channel}
    e_i \text{ flips the syndrome }  \iff \begin{cases}
        e_i > 1 & \text{if } s_i = -1, \\
        e_i < 1 & \text{if } s_i = +1.
    \end{cases}
\end{equation}
To satisfy the conditions in Eq.~\eqref{eq:gauss-noise-to-channel} define the syndrome error rate $p_{\rm syndr}$ as
\begin{equation}\label{eq:p-syndr-from-gauss-noise-def}
    p_{\rm syndr} = \begin{cases}
            \frac{1}{ \sqrt{2 \pi \sigma^2}} \int_{-\infty}^{-1} e^{-x^2/2\sigma^2} \dd{x} & \text{if } s_i=+1,\\[10pt]
            \frac{1}{ \sqrt{2 \pi \sigma^2}} \int_{1}^{\infty} e^{-x^2/2\sigma^2} \dd{x} & \text{if } s_i=-1.
    \end{cases}
\end{equation}
By symmetry of the Gaussian distribution, this gives equivalent error probabilities for both cases, which can be derived readily by substituting $x\mapsto -x$, 
\begin{align}
     \frac{1}{\sqrt{2 \pi \sigma^2}} \int_{-\infty}^{-1} e^{-x^2/2\sigma^2} \dd{x} 
     = \frac{1}{2} \text{Erfc}\left( \frac{1}{\sqrt{2 \sigma^2}} \right),
\end{align}
where $x\mapsto \text{Erfc}(x)$ is the complementary error function.
For given $p_{\rm syndr}$ the solution is 
\begin{align}\label{eq:app-perr-to-sigma}
    \sigma = \frac{\sqrt{2}^{-1}}{\text{Erfc}^{-1}(2 p_{\text{sydr}})} = \frac{\sqrt{2}^{-1}}{\text{Erf}^{-1}(1 - 2 p_{\text{sydr}})},
\end{align}
where $x\mapsto \text{Erfc}^{-1}(x)$ is the inverse of the complementary error function and $\text{Erf}^{-1}(x)$ the inverse of the error function.
This allows us to relate the discrete qubit and (analog) cat qubit error models in a one-to-one correspondence.

\subsection{Belief-propagation}\label{sec:app_bp}
Both the SSMSA decoder proposed in Ref.~\cite{raveendran_soft_2022} and our ATD method are based on \emph{belief propagation} (BP), which is a decoding algorithm that has been adapted from classical (LDPC) codes to quantum codes~\cite{fossorier_soft-decision_1995, fossorier_iterative_2001, panteleev_degenerate_2021}. 
In this section, we briefly review the main aspects of BP using minimum-sum update rules relevant to our methods. 
We refer the reader to the literature in the field for a more in-depth discussion, for instance Refs.~\cite{roffe_decoding_2020, panteleev_degenerate_2021}.
Since we focus on CSS codes that can be seen as a combination of two classical linear codes, we focus on a single check side in the following. 

Given a syndrome $s = H\cdot e$, the objective of the decoder is to find the most likely error $e$. In practice, this amounts to finding a minimum (Hamming) weight estimate $\varepsilon$ for the error, i.e., 
\[ \varepsilon = \text{argmax}_e \mathbf{Pr}(e|s).\]
In an i.i.d.~noise model, $\varepsilon$ can be computed bit-wise by computing the marginal probabilities 
\begin{equation}
\label{eq:bp_marginals}
\mathbf{Pr}(e_i) = \sum_i^n \mathbf{Pr}(e_1, e_2, \dots, \hat{e_i}=1, e_{i+1}, \dots, e_n|s),
\end{equation}
where the hat $\hat{e}_i$ indicates that the variable is left out, i.e., summation over all variables except $e_i$.

The goal of BP is to compute these probabilities in an iterative way by using the natural factorization given by the Tanner graph of the code (also called the \emph{factor graph} in this context). 
The marginals $\mathbf{Pr}(e_i)$ are then used to infer an estimate $\varepsilon$ by setting 
\begin{equation}\label{eq:bp_hard-decisions}
\varepsilon =
    \begin{cases}
      1 & \text{if}\ \mathbf{Pr}(e_i) \geq 0.5, \\
      0 & \text{otherwise}.
    \end{cases}
\end{equation}
Belief propagation computes marginals using an iterative message-passing procedure, where in each iteration, a message is sent from each node to its neighbors.
The messages constitute sets of ``beliefs'' on the probabilities to be computed.
The value of the messages depends on the syndrome and the bit error channel. 
In this work, we use a serial schedule to compute BP marginals. Our implementation is outlined in~\Cref{alg:bp} and is described below.

\begin{algorithm}[t]
\SetAlgoLined
\caption{Hard syndrome belief-propagation (MSA) with serial scheduling \label{alg:bp}}
	$s$: Syndrome\;
 	$H$: Parity-check matrix\;
	$\mathcal{N}(c)$: Bits in the neighborhood of check $c$\;
	$\mathcal{M}(b)$: Checks in the neighborhood of bit $b$\;
	$\mu_{c,b}$: Check-to-bit update from check $c$ to bit $b$\;
	$\nu_{c,b}$: Bit-to-check update from bit $b$ to check $c$\;
	$\lambda_{b}=\log((1-p)/p)$: LLR for bit $b$\;
	$p$: Channel probability\;
	bit-count: The number of bits\;
	max-iter: The maximum no. BP iterations to run\;
	\KwResult {estimate $\varepsilon$}
	\For(\tcp*[h]{Initialization}){all $c,b$ where $H_{c,b}\neq 0$}{ 
		$\nu_{b,c} := \log((1-p)/p)$ \label{alg:bp-initialization} 
	}
	\For(\tcp*[h]{Main iteration loop}){iter to max-iter}{ 
			\For(\tcp*[h]{Serial bit loop}){$b \in \{1 ,\dots , \text{bit-count} \}$}{ 
						$\lambda_b = \log((1-p)/p)$\tcp*{Initialise LLRs}
			\For(\tcp*[h]{Loop over neighbors}){$c \in \mathcal{M}(b)$}{ 
				$|\mu_{c,b}| = \min\limits_{b' \in \mathcal{N}(c) \setminus \{b\}}{\left|\nu_{b',c}\right|}$\;	
				$\text{sgn}(\mu_{c,b}) = \text{sgn}(\lambda_b) \cdot \prod\limits_{b' \in \mathcal{N}(c) \setminus \{b\}}{\text{sgn}(\nu_{c,b'}})$\;
				$\mu_{c,b} = \alpha \cdot \text{sgn}(\mu_{c,b})\cdot \left|\mu_{c,b}\right|$\tcp*{Check to bit}
				$\lambda_b = \lambda_b + \mu_{c,b}$\tcp*{Update LLRs}
			}
			\For{$c \in \mathcal{M}(b)$}{
				$\nu_{b,c} = \lambda_b - \mu_{c,b}$\tcp*{Bit to check}
			}            
			\eIf(\tcp*[h]{Hard decision on bit $b$}){$\lambda_b\leq 0$}{ 
				$\varepsilon_b = 1$\;
			}{ 
				$\varepsilon_b = 0$\;
			}
		}
		\If{$s = H\cdot \varepsilon$}{
			return $\varepsilon$\tcp*{Converged, return estimate.}
		} 
	}
	return $\varepsilon$\tcp*{Return estimate after maximum number of iterations reached.}
\end{algorithm}
In the first step of BP decoding, the values of the bit nodes $b_j$ are initialized to the \emph{log-likelihood ratios} 
(LLRs) $\lambda_i$ of the error channel 
\begin{equation}\label{eq:bp_llrs}
\lambda_i := \log\left(\frac{1-p}{p}\right).    
\end{equation}
In every iteration, each check node $i$ sends messages to neighboring bit nodes $j \in \calN(i)$, denoted as $\mu_{i,j}$. 
The value of the check-to-bit message $\mu_{i,j}$ is computed as a function of the syndrome $\gamma_{i}$ and the incoming bit-to-check messages $v_{i,j'}$ as
\begin{equation}\label{eq:bp_ctb-messages}
\alpha\cdot \sgn(\gamma_i) \cdot \prod\limits_{j' \in \mathcal{N}(i) \setminus \{j\}}{\text{sgn}(\nu_{i,j'}}) \cdot \min\limits_{j' \in \mathcal{N}(i) \setminus \{j\}}{\abs{\nu_{i,j'}}}\rm.
\end{equation}
The factor $\alpha \in \R$ is called the \emph{scaling factor}~\cite{emran_simplified_2014}.
The bit-to check node messages $\nu_{j,i}$ are computed as
\[
\mu_{j,i} := \lambda_{j} + \sum_{i'\in \calN(j) \setminus i} \mu_{i',j}.
\]
It is well-known that BP computes only the marginals $\mathbf{Pr}(e_i)$ exactly on factor graphs that are trees (in a single step). In the more general setting, where the graph contains loops, the computed marginals are approximate~\cite{kschischang_factor_2001}. To check for termination of the iterative procedure, we use the marginals and Eq.~\eqref{eq:bp_hard-decisions} to infer an estimate $\varepsilon$ from the currently computed marginals and check if $\varepsilon$ is valid for the given syndrome. 
i.e.,  if $s = H \cdot \varepsilon$ the solution $\varepsilon$ is valid. If $\varepsilon$ is valid, the BP algorithm terminates and $\varepsilon$ is returned as the decoding estimate.

We use the BP+OSD implementation provided in the LDPC2 package by Roffe~\etal{}available on Github~\url{https://github.com/quantumgizmos/ldpc/tree/ldpc_v2}.

\subsection{Soft-syndrome MSA}\label{sec:app_si-decoder}
The SSMSA algorithm is essentially equivalent to BP as sketched in~\Cref{alg:bp} with some alterations. 
Instead of the hard syndrome vector $s$, the input is an analog syndrome vector $\tilde{s}$, whose corresponding LLR vector is denoted $\gamma$ 
(cf.~Eq.~\eqref{eq:llr_analog_information}). 
The initialization and bit-to-check messages are computed equivalently. 
For computing the check-to-bit messages, the analog syndrome $\gamma$ is taken into account.
If the syndrome value is below a pre-defined cutoff value $\Gamma$ that models the reliability of the syndrome information, the syndrome information is treated as unreliable and the messages $\mu_{c,b}$ are instead computed as
\begin{equation}
\mu_{c,b} := \begin{cases}
    \min\limits_{b' \in \mathcal{N}(c)\setminus b}({\abs{\nu_{b',c}}}) & \text{if}\ |\gamma_c|>\Gamma , \\
    \min\limits_{b' \in \mathcal{N}(c)\setminus b}({\abs{\nu_{b',c}},\abs{\gamma_c}}) & \text{otherwise}.
\end{cases}
\end{equation}
Note that there is a case in which Algorithm~\ref{alg:ssmsa} erases the analog syndrome information. This occurs when the absolute value of the analog syndrome is smaller than the value of all incoming messages of the check node $c$, and the sign of the incoming messages matches the sign of the syndrome. From Line~\ref{line:ssmsa-overwrite} of Algorithm~\ref{alg:ssmsa}, we see that in this case the analog syndrome value is overwritten and thus lost.

\begin{algorithm}[t]
\SetAlgoLined
\caption{Soft-Syndrome MSA \label{alg:ssmsa}}
    $\gamma_i$: Analog syndrome LLR\;
    $H$: Parity-check matrix\;
    $\mathcal{N}(c)$: Bits in the neighborhood of check $c$\;
    $\mathcal{M}(b)$: Checks in the neighborhood of bit $b$\;
    $\mu_{c,b}$: Check-to-bit update from check $b$ to bit $b$\;
    $\nu_{c,b}$: Bit-to-check update from bit $b$ to check $c$\;
    $\lambda_{b}$: LLR for bit $b$\;
    $\Gamma$: Cutoff for soft-info decoding\;
    $p$: Channel probability\;
    bit-count: the number of bits\;
    max-iter: The maximum no. BP iterations\;
    \KwResult{estimate $\varepsilon$}
    \For(\tcp*[h]{Initialization}){all $c,b$ where $H_{c,b}\neq 0$}{ 
    	$\nu_{b,c} := \log((1-p)/p)$\; \label{alg:ssmsa-initialization}    
     }
    \For(\tcp*[h]{Main iteration loop}){iter to max-iter}{
    	\For(\tcp*[h]{Serial bit loop}){$b \in \{1 .. \text{bit-count} \}$}{     
        $\lambda_b = \log((1-p)/p)$\tcp*{Initialise LLRs}
	\For(\tcp*[h]{Loop over check bits}){$c \in \mathcal{M}(b)$}{
            \If(\tcp*[h]{Virtual check update}){$\abs{\gamma_c}\leq\Gamma$}{ 
    		  $|\mu_{c,b}| = \min\limits_{b' \in \mathcal{N}(c)}({\abs{\nu_{b',c}},\abs{\gamma_c}})$\;
    		\If{$\abs{\gamma_c} < \min\limits_{b' \in \mathcal{N}(c)}({\abs{\nu_{b',c}}})$}{
    		      \eIf{$\text{sgn}(\gamma_c) == \prod\limits_{b' \in \mathcal{N}(c)}{\text{sgn}(\nu_{b',c}})$}{
    		          $\gamma_c = \text{sgn}({\gamma_c}) \cdot \min\limits_{b' \in \mathcal{N}(c)}({\abs{\nu_{b',c}}})$\;\label{line:ssmsa-overwrite}
    		      }{ 
    		          $\gamma_c=-1\times \gamma_c$\;
    		      }	
    		}	
		}\Else(\tcp*[h]{Default to MSA if above cutoff}){
		$|\mu_{c,b}| = \min\limits_{b' \in \mathcal{N}(c) \setminus \{b\}}{\abs{\nu_{b',c}}}$\;     
            }
	
        $\text{sgn}(\mu_{c,b}) = \text{sgn}(\gamma_c) \cdot \prod\limits_{b' \in \mathcal{N}(c) \setminus \{b\}}{\text{sgn}(\nu_{b',c}})$\;
        
        $\mu_{c,b} = \text{sgn}(|\mu_{c,b}|)\cdot \abs{\mu_{c,b}}$\tcp*{Check to bit}
        
        $\lambda_b = \lambda_b + \mu_{c,b}$\tcp*{Update LLRs}
        }
    
    \For{$c \in \mathcal{M}(b)$}{
        $\nu_{b,c} = \lambda_b - \mu_{c,b}$\tcp*{Bit to check}
    }
    \eIf(\tcp*[h]{Hard decision on bit $b$}){$\lambda_b\leq 0$}{ 
        $\varepsilon_b = 1$\;
    }{ 
        $\varepsilon_b = 0$\;
    }
    
    }
    $s_c = \text{sgn}{(\gamma_c)}$\tcp*{Hard syndrome}
    \If{$s = H\cdot \varepsilon$}{
        return $\varepsilon$\tcp*{Converged, return estimate.}
    }
 }
return $\varepsilon$\tcp*{Return estimate after maximum number of iterations reached.}
\end{algorithm}

Our implementation of the SSMSA decoder is made publicly available in the LDPC2 package~\url{https://github.com/quantumgizmos/ldpc/tree/ldpc_v2}.

\subsection{Analog Tanner graph 
decoder}\label{sec:app_atd-details}
In \emph{analog Tanner graph decoding} (ATD), we use the Tanner graph of the code $\calT$ to construct the \emph{analog Tanner graph} (ATG). This allows us to directly incorporate the analog syndrome information in \emph{virtual nodes} in the Factor graph. 

In the initialization phase of the ATD decoder, BP sets the values $\lambda_i, i\in [n]$ of all bit nodes in the factor graph to 
\begin{equation}\label{eq:bp_bit-llrs}
 \lambda_i = \log \left(\frac{1-p}{p}\right), i\in[n].
\end{equation}
To ensure that we initialize the values of the analog nodes $\lambda_{n+j}, j\in[m]$ with the value of the analog syndrome $\gamma_j$, we derive the error channel probabilities $p'_i$ from Eq.~\eqref{eq:bp_bit-llrs}
\begin{align}
    p'_i & = \frac{1}{e^{\gamma_i}+1}.
\end{align}
Thus we set 
\begin{equation}
    p'_{n+j} = \frac{1}{e^{\gamma_j}+1}, j\in[m]  
\end{equation}
for the analog nodes to ensure that after the initialization phase of BP the bit nodes are initialized with the LLRs, and the virtual nodes with the analog syndrome (i.e., the LLR) values.

\section{Obtaining analog information for data qubits and concatenated GKP codes}\label{sec:bias-pres-syndromes-cats}

One might raise the question as to whether it is possible to include more analog information in the decoding graph, e.g., by considering analog values associated with qubits (data nodes in the factor graph used for decoding) as well.
The answer to this question is positive under the assumption that one performs active error correction on the bosonic qubit. 
In the considered phenomenological noise model that is inspired by stabilized cat qubits, this is not the case as the cat qubit is autonomously protected by the engineered stabilization mechanism as discussed in~\secref{sec:cats}.
We emphasize that in a gate-based scheme, it is not possible to perform active error correction on cat qubits, i.e., measure their stabilizers, as cat codes cannot be described as a stabilizer code in the conventional sense.
Fundamentally, this is due to the non-existence of a phase operator in the harmonic oscillator Hilbert space.

Although in some cases, for example, dissipative stabilized cat qubits, information about certain types of errors can also be obtained by continuously monitoring the buffer mode that is used to implement the dissipation mechanism~\cite{gautier_designing_2023}, doing so with sufficiently high reliability seems to be a task of similar complexity as implementing a QEC protocol.
Another way of incorporating qubit analog information into the decoding graph requires a changing the computing paradigm from a gate-based scheme to the measurement-based scheme in which qubits are explicitly measured.
The measured value will be real valued and its magnitude can be used to assign LLR to that specific qubit.

However, if the bosonic code is actively corrected in some way, one can use this information in the decoder as well. A straightforward example 
is the single mode GKP code as it is a stabilizer code with generators given by the displacements
in phase space $\mathbb{R}^2$ 
\begin{equation}
\langle S_X := e^{2 i \sqrt{\pi} \hat{p}}, S_Z := e^{-2 i \sqrt{\pi} \hat{x}} \rangle, 
\end{equation}
where $\hat{p}$ and $\hat{x}$ are the momentum and position operators of the harmonic oscillators.
Here, 
\begin{equation}
	e^{2 i \sqrt{\pi} \hat{p}} = D(- \sqrt{2}(1,0)^T), \,  e^{-2 i \sqrt{\pi} \hat{x}} = D(\sqrt{2}(0,1)^T)
\end{equation}
with 
\begin{equation}
D(\xi) := \exp( - i \sqrt{2\pi} \xi^T (\hat{p},-\hat{x})^T) 
\end{equation}
being the single mode shift operator, generating shifts in phase space by $\xi\in \mathbb{R}^2$.
One can measure the stabilizer generators using Steane-type error correction which requires auxiliary GKP qubits~\cite{noh_fault-tolerant_2020}. 
Assuming the availability of noiseless auxiliary qubits, the Steane-type error correction determines 
the shift error the data qubit has undergone modulo %
This means that shift errors up to a magnitude of at most $\sqrt{\pi}/2$ can be corrected, while shift errors that have a larger magnitude typically lead to logical errors in the GKP qubit subspace.
Suppose that the measurement yields an outcome $x_m$ for the $x$-quadrature shifts. 
Then, if 
$\abs{x_m \mod{\sqrt{2\pi}} }\approx 0$
it is unlikely that this qubit has undergone a logical GKP error if we assume 
as above 
that shift errors follow a Gaussian distribution $x_m \sim \mathcal{N}(0, \sigma^2)$ with mean zero and variance $\sigma^2$.
However, if 
$\abs{x_m \mod{\sqrt{2\pi}} }\approx \sqrt{\pi} / 2$,
a logical GKP qubit error is significantly more likely. 
Thus, it is possible to quantify this likelihood and use it to also initialize the data nodes of the decoding graph [black circles in \sfigref{fig:fig1}{d-f}] with analog information and then apply ATD to decode.
See also \figref{fig:gkp_analog_information} for a visualization of the likelihood function in analogy to \sfigref{fig:fig1}{c}.

\begin{figure}[t]
    \centering
    \includegraphics{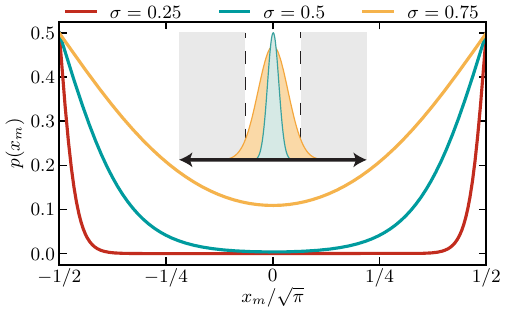}
    \caption{Obtaining analog information from GKP data qubits. The main panel shows the GKP qubit error probability conditioned on the measurement outcome $x_m$ for different values of noise strength $\sigma$. Detecting a value $x_m$ close to zero corresponds to a small error probability for any value of $\sigma$, while the error probability for measurement outcomes closer to the ``decision boundaries'' $\pm \sqrt{\pi} / 2$ are highly dependent on the assumed noise channel and its variance $\sigma^2$. The inset shows a sketch of the distribution of measurement outcomes $x_m$ for two different values of $\sigma$ and the dashed lines indicate the decision boundaries.}
    \label{fig:gkp_analog_information}
\end{figure}

\section{Concatenated multi-mode GKP
and rotation symmetric bosonic
codes}\label{sec:GKP}

In a measurement-based setting, analog information can be extracted through teleportation-based Knill-type error correction~\cite{knill_fault-tolerant_2004}. This works also for \emph{multi-mode instances of the GKP code}~\cite{noh_low-overhead_2022, larsen_fault-tolerant_2021, conrad_gottesman-kitaev-preskill_2022}. Generally, one can think of multi-mode encodings
beyond single-mode encodings in terms of displacement operators
\begin{equation}
	D(\xi) := \exp{ - i \sqrt{2\pi} \xi^T J (\hat{p}_1,\dots,\hat{p}_m ,-\hat{x}_1,\dots,-\hat{x}_m)^T}
\end{equation}
in phase space $\R^{2m\times 2m}$ of $m$ 
bosonic modes, with
\begin{equation}
	J := \left[
	\begin{array}{cc}
	0 & \Id_m\\
	- \Id_m & 0
	\end{array}
	\right]
\end{equation}
Then one can generally define a stabilizer group of a GKP code in terms of displacements
\begin{equation}
	\langle D(\xi_1),\dots, D({\xi_{2m}})
    \rangle,
\end{equation}
where $\xi_1,\dots, \xi_{2m}\in \R^{2m\times 2m}$ are linearly independent and we have that
$\xi_i^T J \xi_j \in \Z$ for all $i,j$. This 
defines a stabilizer group isomorphic to a lattice
\cite{conrad_gottesman-kitaev-preskill_2022}. In this way, one obtains multi-mode information in the syndrome
measurements. For this reason, the methods introduced here also contribute to the question of how to decode GKP codes.
In a similar way, one can consider large classes of
alternative bosonic codes known as \emph{rotation symmetric bosonic codes}~\cite{grimsmo_quantum_2020, hillmann_performance_2022}. 
Both 
our decoding techniques and software tools are easily adapted to incorporate analog information for the data qubits, as explained above.

\section{Additional results} \label{app:additional_results}
In this section, we present additional simulation results for various parameter settings of the considered decoder implementations. We also present results for the phenomenological threshold of the three-dimensional toric code (i.e., the 3DSC with period boundaries).

\subsection{Phenomenological noise threshold of the three-dimensional toric code}
\begin{figure}[!b]
    \centering
    \includegraphics{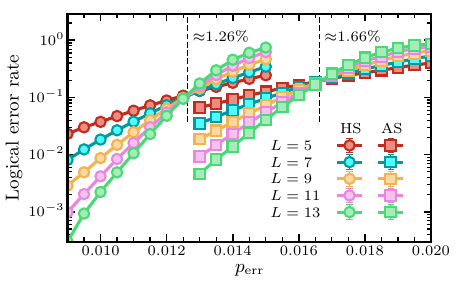}
    \caption{Phenomenological noise threshold of the three-dimensional toric code. The logical error rates are obtained by simulating a pure bit-flip noise model with $p_{\rm err}^{\rm data} = p_{\rm err}^{\rm synd}$.
    The syndrome is collected in $2 L$ measurement rounds of which the last one is noiseless. We decode using minimum weight perfect matching that yields a threshold at $\approx 1.26 \%$ and $\approx 1.66 \%$ using the \emph{hard syndrome} (HS) and the \emph{analog syndrome} (AS), respectively.
    }
    \label{fig:3dtc-th-nonss-side}
\end{figure}

As a by-product of our QSS simulations, we obtain a threshold of the non-single-shot side of the 3DTC under phenomenological noise of $\approx 1.26\%$  and $\approx 1.66 \%$ using \emph{hard syndromes} (HS) and \emph{analog syndromes} (AS), respectively. The corresponding threshold plots are shown in~\figref{fig:3dtc-th-nonss-side}.
This generalizes recent code capacity results presented in Ref.~\cite{huang_tailoring_2022}.
Here, we simulated a pure bit-flip noise model  with $p_{\rm err}^{\rm data} = p_{\rm err}^{\rm synd}$. Measurement results are collected over $2L$ rounds and then decoded using minimum-weight perfect-matching.
For the case of hard syndromes, the noise from the syndrome is added as discussed in~\secref{sec:syndr-noise-model}. This ensures the effective error channel for hard syndromes is the same as that for analog syndromes. 
We use the PyMatching implementation of minimum-weight-perfect matching  ~\cite{higgott_pymatching_2022, higgott_sparse_2023} as the decoder for these simulations.

\subsection{BP parameter optimization}\label{sec:app_bp-paremeter-opt}
Since there are various parameters that allow one to fine-tune the BP+OSD implementation, we conducted a set of numerical experiments to determine which parameter setting performs the best in the scenario considered. 
However, note that for this work we focus on techniques and methods rather than low-level optimizations in general. 
Hence, the presented numerical results are generally implementation-dependent and most likely prone to further optimization and fine-tuning.

In summary, the conducted numerical experiments
demonstrate that there can be quite significant differences in decoding performance (achieved logical error rate/threshold and
number of BP iterations/elapsed time) depending on the BP+OSD parameters, most notably the chosen BP scaling factor $\alpha$ (cf.~\secref{sec:app_bp}) and the OSD method. 

The main findings are that OSD-cs with a scaling factor $\alpha=0.5$ to $\alpha=0.6$ perform best.
Furthermore, concerning the SSMSA implementation, we find that a cutoff of $\Gamma=5$ performs the best for the considered lifted product code family. 
For readers interested in more detailed results, we refer to the GitHub repository~\url{github.com/cda-tum/mqt-qecc}, where we present detailed simulation results.

\section{Numerical simulation details}\label{sec:numeric-sim-details}

In this section, we discuss the implementation details of numerical simulations and the main techniques used.
Since we focus on CSS codes using depolarizing (biased) noise without correlations, we decode \qgate{X} and \qgate{Z} errors  separately (that implies that a \qgate{Y} error on a qubit $q$ is treated as \qgate{X} and \qgate{Z} error on $q$). 
To determine the \emph{logical error rate} of an $[[n,k,d]]$ single-shot code, decoder combination in the considered phenomenological (cat qubit) noise model, we use the following procedure: 
\begin{enumerate}
	\item Sample an error vector $e \in \Ft^n$ (bit-wise and dependent on the error channel.
	\item Compute the syndrome $s = H\cdot e$.
	\item Sample and apply a syndrome error $e_s \in \Ft^m$ to obtain the noisy syndrome $\hat{s}$. 
	\item Apply (analog) single-stage decoding to obtain an estimate $\varepsilon$.
	\item Check if $r = \varepsilon + e$ is a logical.
\end{enumerate}
The sample run is successful if $r$ is not a logical operator (i.e., the correction induced a stabilizer, which does not alter the encoded logical information). 

We apply a similar procedure to estimate the logical error rate for a code and decoder combination when applying the (analog) overlapping window method to decode over time with multiple syndrome measurements. 
To simulate decoding in $R$ rounds of noisy syndrome measurement, we first compute $R$ noisy syndromes and then apply overlapping window decoding as described in the main text to obtain a single sample. 
We repeat this procedure for a maximum number of samples $N$. 

The logical \qgate{X}/\qgate{Z} error rate $p_\ell^{X/Z}$ is the fraction of failed runs $n_{\rm fail}$ for $N$ total samples
\begin{equation}\label{eq:ler}
    p_\ell^{X/Z} := \frac{n_{\rm fail}}{N} .
\end{equation}
The error bars $\mathit{e_{X/Z}}$ are computed by
\begin{equation}
	e_{X/Z}^2 = {(1-p_\ell) \frac{p_\ell}{N}}.
\end{equation}
The simulations are terminated if the maximum number of samples $N$ is reached, or if the errors fall below a certain precision cutoff determined by the ratio of the error bar value and the logical error rate, which we set to $10^{-1}$.

To give a better comparison of logical error rates for codes that encode a different number of logical qubits $k$, we use the \emph{word error rate} (WER) for codes with $k>1$, which intuitively can be understood as logical error rate per logical qubit. 
The WER $p_{w}$ is computed as
\begin{equation}
	p_w = 1- (1-p_\ell)^{1/k-1}.
\end{equation}
We can use the methods described above to obtain a \emph{threshold} $p_{\rm th}$ estimation by computing the logical error rates for code instances of a code family with increasing distance and increasing physical error rate and then estimating where the graphs cross.
We use a standard approach based on a finite-size scaling regression analysis~\cite{wang_confinement-higgs_2003, huang_tailoring_2022, chubb_statistical_2021}.

To this end, we perform a quadratic fit on the logical error rate data obtained by numerical experiments.
Let $p_{\rm th}$ be the threshold physical error rate we want to determine, 
$\mu>0$ a parameter called the \emph{critical exponent}, and define the \emph{rescaled physical error rate}
\begin{equation}
    x := (p-p_{\rm th})\cdot d^{1/\mu},
\end{equation}
which is $x=0$ at $p=p_{\rm th}$.
We use $x$ to fit the simulation data to the following quadratic ansatz $\Phi(p,d)$
\begin{equation}
    \Phi(p, d) := ax^2 + bx + c,
\end{equation}
where $a, b$, and $c$ are coefficients of the quadratic ansatz and are parameters to be determined by fitting the data. 
Note that $\Phi(p,d)$ is only a valid approximation near $p=p_{\rm th}$, therefore only data points close to this region were used to compute the fit.
Given this ansatz, we use the logical error rates $p_\ell$ (cf. Eq.~\eqref{eq:ler}) to obtain the free parameters $(p_{\rm th}, \mu, a,b,c)$ by computing a fit using the minimized mean square error.

Note that, since we focus on CSS codes, where the \qgate{X} and \qgate{Z} decoding is done separately (assuming non-correlated errors), we need to compute the \emph{combined} logical error rates from the separate experimental data. 
Given a \qgate{X} and \qgate{Z} logical error rates $p_{\ell}^\qgate{X}, p_{\ell}^\qgate{Z}$ from numerical experiments, we compute the combined logical error rate $p_\ell$ as 
\begin{equation}
    p_{\ell} := p_{\ell}^\qgate{X} \cdot (1-p_{\ell}^\qgate{Z}) + p_{\ell}^\qgate{Z} \cdot (1-p_{\ell}^\qgate{X}) + p_{\ell}^\qgate{X} \cdot p_{\ell}^\qgate{Z}.
\end{equation}
Moreover, the corresponding errors are computed using standard methods for \emph{propagation of uncertainty}.
i.e.,  we use the variance formula~\cite{ku_notes_1966}
\begin{equation}
    s_f^2:= {\left(\frac{\bdry f}{\bdry x}\right)^2 s^2_x + \left(\frac{\bdry f}{\bdry y}\right)^2 s^2_y + \left(\frac{\bdry f}{\bdry z}\right)^2 s^2_z+ \cdots},
\end{equation}
where $s_f$ is the standard deviation of $f$, and $s_k$ the standard deviation of $k\in\set{x,y,z,\dots }$.
To be concrete, we apply the formula above to the separately computed errors $e_{x}, e_{z}$, to obtain the overall error $e$ with
\begin{equation}
    e^2 := {e_x^2 \cdot (1-e_z)^2 + e_z^2 \cdot (1-e_x)^2}.
\end{equation}
\medskip

\section{Open-source software}\label{sec:sw-framework}
All our techniques are available in the form of open-source software available on GitHub~\url{https://github.com/cda-tum/mqt-qecc} as part of the \emph{Munich Quantum Toolkit} (MQT) and partly in the LDPC2 package~\url{https://github.com/quantumgizmos/ldpc/tree/ldpc_v2}.
With the proposed software, we hope to provide a set of useful tools for analog information decoding and to emphasize the need for open-source implementations to foster public review, reproducibility, and extendability.

\bibliography{references}
\end{document}